\documentclass[conference]{IEEEtran}
\usepackage{algorithm}
\usepackage{algpseudocode}
\usepackage{algorithmicx}
\usepackage{amsmath}
\usepackage{amsthm}
\usepackage{amssymb}

\usepackage{ulem}

\newtheorem{thm}{Theorem}

\newtheorem{prop}{Proposition}

\newtheorem{rmk}{Remark}

\usepackage{tikz}
\usepackage{amsmath}

\ifCLASSINFOpdf
\else
\fi
\begin{document}
%
\title{AnonPSI: An Anonymity Assessment Framework
for PSI}

\author{\IEEEauthorblockN{Bo Jiang}
\IEEEauthorblockA{TikTok Inc.\\
bojiang@tiktok.com}
\and
\IEEEauthorblockN{Jian Du}
\IEEEauthorblockA{TikTok Inc.\\
jian.du@tiktok.com}
\and
\IEEEauthorblockN{Qiang Yan}
\IEEEauthorblockA{TikTok Inc.\\
yanqiang.mr@tiktok.com}}


%


\IEEEoverridecommandlockouts
\makeatletter\def\@IEEEpubidpullup{6.5\baselineskip}\makeatother
\IEEEpubid{\parbox{\columnwidth}{
    Network and Distributed System Security (NDSS) Symposium 2024\\
    26 February - 1 March 2024, San Diego, CA, USA\\
    ISBN 1-891562-93-2\\
    https://dx.doi.org/10.14722/ndss.2024.241279\\
    www.ndss-symposium.org
}
\hspace{\columnsep}\makebox[\columnwidth]{}}

\maketitle

\begin{abstract}

Private Set Intersection (PSI) is a widely used protocol that enables two parties to securely compute a function over the intersected part of their shared datasets and has been a significant research focus over the years. However, recent studies have highlighted its vulnerability to Set Membership Inference Attacks (SMIA), where an adversary might deduce an individual's membership by invoking multiple PSI protocols. This presents a considerable risk, even in the most stringent versions of PSI, which only return the cardinality of the intersection. This paper explores the evaluation of anonymity within the PSI context.
Initially, we highlight the reasons why existing works fall short in measuring privacy leakage, and subsequently propose two attack strategies that address these deficiencies. Furthermore, we provide theoretical guarantees on the performance of our proposed methods.
In addition to these, we illustrate how the integration of auxiliary information, such as  {the sum of payloads associated with members of the intersection (PSI-SUM),} can enhance attack efficiency.
We conducted a comprehensive performance evaluation of various attack strategies proposed utilizing two real datasets. Our findings indicate that the methods we propose markedly enhance attack efficiency when contrasted with previous research endeavors.  {The effective attacking implies that depending solely on existing PSI protocols may not provide an adequate level of privacy assurance.  It is recommended to combine privacy-enhancing technologies synergistically to enhance privacy protection even further.}
\end{abstract}

\section{Introduction}

Private Set Intersection (PSI) protocols allow two parties to securely compute a function over the intersection of their datasets without directly revealing the intersection itself   \cite{PSI_agrawal,PSI_de,PSI_fenske} and have gained significant interest in the industry. Sensitive information within the intersection, such as individual identities, attributes, and memberships at the other party, remains unrevealed.

 {While the intersection set is not directly disclosed, most two-party secure computation protocols disclose the intersection size may inadvertently leak membership information \cite{281334}.}
PSI-CA (Private Set Intersection-Cardinality) is one such protocol that releases the cardinality of the intersection to one party  \cite{PSI_CA_Freedman,jofc-2016-28434}. This protocol serves as an essential building block for various applications, such as new friend recommendations in online social networks \cite{10.1145/3267323.3268965} and contagious disease tracking \cite{10.1007/978-3-030-64840-4_29}. Another popular PSI protocol is PSI-SUM \cite{PSI_SUM_Google, cryptoeprint:2020/385}. In addition to the cardinality of the intersection,  {PSI-SUM also reveals the summation of payloads associated with members of the intersection.  One application of PSI-SUM is the secure computation for advertising measurement between an advertiser and an ad provider.
Similar PSI protocols including  Private-ID \cite{cryptoeprint:2020/599} reveal the intersection size. } 



Recent research indicates that even the disclosure of the intersection cardinality can introduce vulnerabilities \cite{281334}. Intuitively, cardinality is not independent of the identities of individuals in the intersected set, especially when the size is small. Therefore, these intersection size revealing protocols enable an adversary to infer whether a targeted individual is in the intersection or not, which is also known as the membership information. Notably, individuals who intersect with the other party's set are classified as positive members, while those outside the intersection are classified as negative members. Depending on the application context, either or both of these positive/negative memberships can be sensitive from a privacy perspective as shown below.
 {Consider, for example, a situation involving COVID-19 testing. Individuals who have tested positive are treated as sensitive information. Another scenario arises when politicians seek to advance their positions and persuade people to support a bill. They may examine existing datasets to identify those who abstained from voting. In this context, these non-voting individuals, representing negative members, are deemed sensitive.
In a third scenario, such as advertising measurement \cite{PSI_SUM_Google, cryptoeprint:2020/385}, the advertising company would be interested in both positive and negative memberships, which refer to an advertising conversion and non-advertising conversion, respectively. The motivation to determine user memberships serves two primary purposes: firstly, to tally positive members for the targeted advertising campaign, and secondly, to engage with negative members through a distinct advertising campaign.}

A more efficient attacker can even deduce a subset of individuals' membership through a small number of protocol invocations. Such an attack is also known as Set Membership Inference Attack (SMIA). In certain scenarios, the inference of such information may lead to unintended consequences, such as discrimination, targeted advertising, or even social engineering attacks. 
An SMIA is initiated by the attacker being one party participating in the PSI running protocol, through the continual submission of varying input sets. Upon observing the outputs, the adversary can infer the membership information of targeted elements. 

On the other hand, an effective SMIA is useful in measuring the privacy protection guarantee provided by different PSI protocols and providing guidance in designing better and stricter protocols. 
The effectiveness of the SMIA is typically determined by the number of individuals' memberships inferred. Beyond the attack efficiency, to assess the protection provided by PSI, the attack algorithm should also be feasible and valid under certain counterattack measures embedded in the PSI protocol.
For instance, many companies set daily limits for database queries from the public. For example, it allows only 15 queries per day \cite{Rogers2020AMF}. Consequently, a feasible attack should be designed to be sufficiently efficient with a limited number of protocol invocation times.
Moreover, strict PSI protocols are often fortified with privacy-preserving mechanisms, such as noise-adding mechanisms based on differential privacy \cite{Dwork2006,Dwork2008}. The attack should continue to function effectively against such noise with high confidence.

In this paper, our principal contributions can be summarized in three key areas:

\begin{itemize}
    \item We introduce two strategies for implementing membership inference attacks: deterministic and statistical, which take individuals' membership as deterministic values and random variables respectively. Efficient attack algorithms are proposed under each strategy, and we theoretically illustrate the performance guarantee of each (performance lower bound). 
    \item We exhibit the potential of employing auxiliary information to further augment the attack's effectiveness. In this regard, an efficient algorithm specifically tailored for PSI-SUM attacks is proposed  {for the first time.} We amalgamate the attack algorithm in PSI-CA with an offline summation matching algorithm, N-SUM, thereby significantly reducing combination possibilities and enhancing attack efficiency.
    \item We assess the proposed algorithms using two real-world datasets: COVID-tracking systems and TaoBao advertisements. We rigorously test the performance and effectiveness of these algorithms against these datasets, substantiating their practical applicability.  {This further implies that the PSI protocol alone may not ensure adequate privacy. It's advisable to integrate other privacy-enhancing technologies for strengthened privacy safeguards.}
\end{itemize}


\section {Related works}

The concept of Membership Inference Attack (MIA) was introduced by Shokri et al. in 2017  \cite{7958568}, marking a significant milestone in the field of machine learning security. In an MIA, attackers leverage machine learning model outputs such as prediction scores or confidence intervals to deduce whether a specific data point was part of the training set. Intriguingly, these attackers don't require direct access to the model's inner workings or training data; the model's query responses suffice, making these attacks practicable in real-world scenarios where models are often deployed as services. Nasr et al.  \cite{8835245} extended these attacks to collaborative deep learning models, underscoring potential risks in the burgeoning field of federated learning. Theoretical frameworks for comprehending and counteracting these attacks were furnished by Yeom et al.  \cite{8429311}. It's noteworthy, however, that MIAs typically infer a single individual's membership through one observation, while SMIA reasons multiple individual memberships through a sequence of observations. MIAs generally train dedicated models to emulate the target model's behavior, while in SMIA, the attacker alters the input for each protocol call and interprets the results from observations. Defense mechanisms against MIAs, such as adding noise to the model's outputs  \cite{Huang2020InstaHideIS} or employing differential privacy mechanisms  \cite{10.1145/2976749.2978318}, remain viable for SMIA.

A distinct but related line of research pertains to linkage attacks, which attempt to re-identify anonymized data by correlating it with other accessible datasets. Introduced by Latanya Sweeney in 2000  \cite{37926}, linkage attacks gained notoriety when Sweeney demonstrated the possibility of re-identifying individuals in anonymized medical data using publicly available voter registration lists. The rise of big data and computational power has amplified the prominence of linkage attacks. Notably, in 2006, Narayanan and Shmatikov  \cite{4531148} de-anonymized Netflix's movie rating data by linking it with publicly available IMDb ratings. The emerging ubiquity of linkage attacks has underscored the need for robust anonymization techniques and highlighted the limitations of approaches like k-anonymity  \cite{k-anon}. SMIA in contexts such as PSI bears similarities to linkage attacks, as it makes inferences using data from a dependent dataset. However, while linkage attacks typically focus on re-identifying individuals, SMIAs target membership information. In this paper, we present an attack algorithm for PSI-SUM that leverages side information, akin to a linkage attack.

Recently, Guo et al.  in  \cite{281334} proposed an attack that enhances the basic ``toy attack" (submit one element at a time to the PSI protocol, and observe the membership) efficiency from $\mathcal{O}(N)$ to $\mathcal{O}(\log{N})$ (protocol invocation times to infer $N$ individuals' membership). Nevertheless, this improved attack possesses several limitations:
\begin{enumerate}
    \item PSI protocols often cap invocation frequencies within a specific number, necessitating highly efficient iteration-based attacks to garner enough information during the limited protocol calls. The method outlined in  \cite{281334} is inapplicable when the allowable protocol call budget is constrained.
    \item When the PSI incorporates privacy-preserving mechanisms, such as differential privacy, the method in  \cite{281334} becomes impracticable due to its inability to cope with the resultant uncertainty in the attack process.
    \item The aforementioned method focuses on one PSI scenario (PSI-CA), disregarding the side information that other applications like PSI-SUM can provide. This side information could prove instrumental during an attack.
In light of these limitations, our work aims to devise a more effective and efficient attack strategy that not only considers multiple PSI scenarios but is also capable of handling the uncertainties introduced by privacy-preserving mechanisms. 
\end{enumerate}





\section{Preliminaries}

To provide a formal representation of the SMIA for PSI protocols, we define the notations, attack model, and parties involved in the PSI protocol execution. Then we introduce some baseline algorithms for SMIA.


\subsection{Problem formulation}
In a two-party setting, let us represent the dataset of Party 1 as $\mathbb{X} = \{x_1, ..., x_{|\mathbb X|}\}$ and the dataset of Party 2 as $\mathbb{Y} = \{y_1, ..., y_{|\mathbb Y|}\}$. Here,  {we use $|\cdot|$ to denote} the cardinality of a dataset. For instance, $|\mathbb{X}|$ signifies the total number of elements in Party 1's dataset. The cardinality of the intersection between the two datasets is represented by $|\mathbb{X} \cap \mathbb{Y}|$.


 In the context of intersection-size-revealing protocols, the party that receives the intersection size obtains a measure of similarity between its own dataset and the other party's dataset. Given that the party has the freedom to select its own dataset for the protocol, it can strategically assess the other party's dataset according to its interests. A relevant question arises when the party has a multitude of target elements: Can the party determine the membership status of these target elements in the other party's dataset by taking advantage of its ability to measure similarity?

 \textbf{Threat Model:}  We assume the attacker is from Party 1, i.e., $\mathbb X$ belongs to the attacker, and the attacker possesses the following capabilities.
The attacker can participate in multiple intersection-size-revealing protocols as a party. During each protocol execution, the attacker can select its input and obtain the intersection size (and some side information if available, such as the summation in PSI-SUM) resulting from its input and $\mathbb {Y}$. In practical scenarios, the number of times the protocol can be called might be limited due to time constraints or rate-limiting mechanisms. Therefore,  it is assumed that the attacker is permitted to repeatedly engage in protocol invocations under a query budget $\tau$.

Depending on different purposes, the attacker may be interested in different types of membership information, as we discussed in the introduction. Here we assume the adversary is interested in both the positive members and the negative members. To this end, the adversary adaptively designs his attack strategy to maximize membership information leakage.


\textbf{Attacking Strategy: } We consider two types of attack strategies: 
\begin{itemize}
    \item Deterministic attack, where the adversary treats each person's membership as a deterministic value, and his attack is based on narrowing down group size iteratively and eventually landing on selected subsets which incurs a return with a result equal to $0$ (all negative) or the size of the subset (all positive). The attack may re-identify only a small subset of the victim set under small $\tau$, but the inferred memberships are accurate and deterministic.
    \item Statistical attack, where the attacker takes each individual's membership as a binary random variable. His attack is based on updating the posterior belief on these random variables to either maximize (reasoning positive membership) or minimize (reasoning negative membership) them until the stopping criterion is matched. This type of attack strategy will make a guess of each membership in the dataset, and his strategy will guarantee his guessing accuracy above a tolerable threshold.
\end{itemize}
\begin{figure}[t]
\begin{small}
\centering 
{ \includegraphics[width=0.4\textwidth]{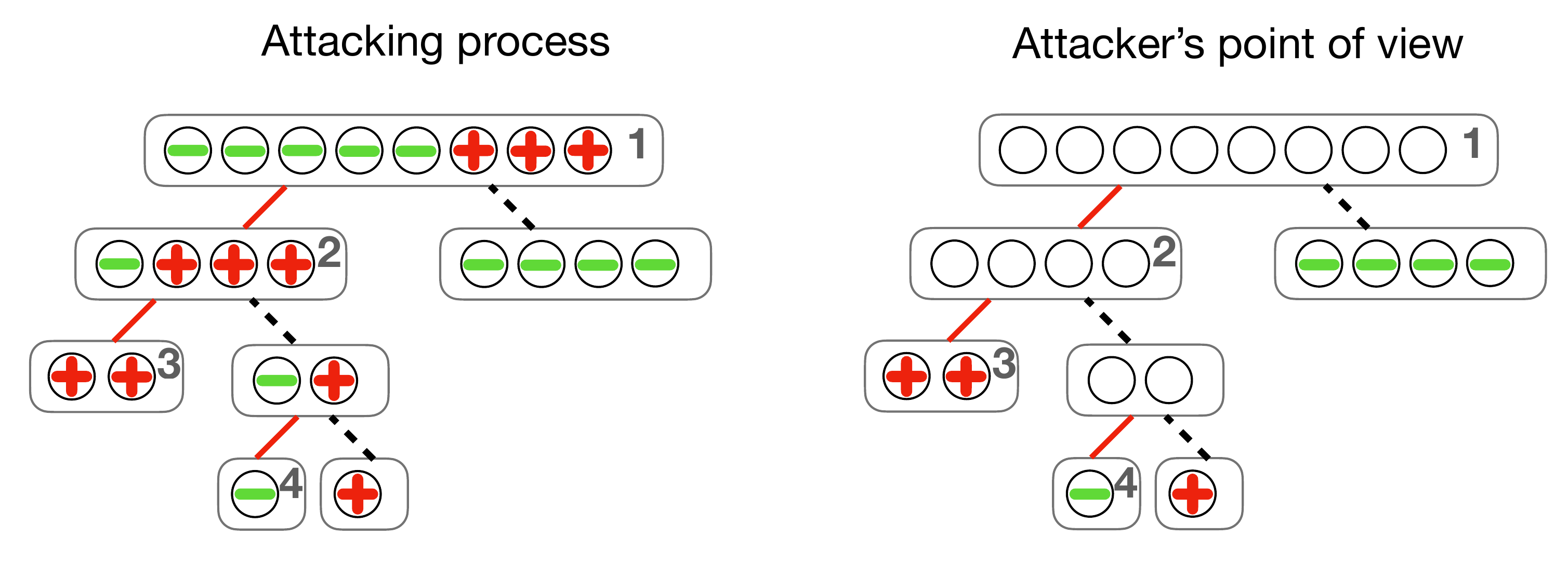} } 
\caption{An Illustration of deterministic hierarchical set membership inference attack.  { We use "+" to denote positive memberships, i.e., individuals in the intersection, and "-" to denote negative memberships, i.e., individuals in the attacker' set only.}} 
\label{det_att} 
\end{small}
\end{figure}
 \subsection{Preliminary attack algorithms}

A brute-force attack involves the attacker randomly selecting an individual from $\mathbb X$ set and checking their membership by invoking the protocol. The attacker then moves on to another individual in $\mathbb X$. To determine the membership of $|\mathbb{X}|$ individuals, the attacker must make $|\mathbb{X}|$ protocol calls. This approach is inefficient for two reasons: first, it requires a large number of calls; and second, it can be easily countered by incorporating some protocol output policies. Such as Apple's PSI-AD protocol, which only returns an intersection size greater than a certain threshold.

The algorithm proposed in  \cite{281334} introduces an attacker that uses a strategy similar to binary search. It constructs a binary tree with target elements at the root, and each non-leaf node's set is divided into two non-empty, disjoint subsets stored in its child nodes. Here node denotes a subset of $\mathbb X$, and a child node is a subset of its parent node. The attacker then performs a depth-first search (DFS) on the tree, invoking the protocol with the victim to get the intersection size for each visited node's set. The search ends at a node where the stored set's size equals the received intersection size. Unvisited subtrees are queued for future DFS.

The attack in  \cite{281334} is also optimized by adopting a greedy DFS approach. The attacker should prioritize searching subtrees with the highest merge probability. This probability is positively related to the ratio of the intersection size of the subtree's root to the number of target elements in the subtree. The attacker can use this ratio as a priority score to sort the subtrees in a priority queue. By doing so, the attacker will first run DFS on the subtree with the highest merge probability, further improving the efficiency of the attack. An illustrative example of this algorithm is shown in Fig. \ref{det_att}.

The attack in  \cite{281334} is more efficient than the toy attack due to its binary-search-like strategy. By constructing a binary tree, the attacker can save at least half of the protocol invocations for every non-root layer. This is because after visiting a child node, the return of the other child node can be deduced from the parent node and the visited child node. The efficiency of the baseline attack is no worse than the toy attack, and it can even be better since it may terminate early when the current node contains all positive or negative members, reducing the number of invocations.

\section{Deterministic attack}

In this section, we propose attacks for the PSI protocols treating each individual's membership as deterministic values. We first illustrate with an example showing the limitations of the aforementioned works. Then we present our DyPathBlazer, a dynamic programming approach for PSI-CA protocol. At the end of this section, we show how  {auxiliary information in the output}  improves attack efficiency and propose TreeSumExploerer, an efficient attack algorithm for PSI-SUM.


\subsection{Even v.s. uneven tree partition}

\begin{figure}[t]
\begin{small}
\centering 
{ \includegraphics[width=0.45\textwidth]{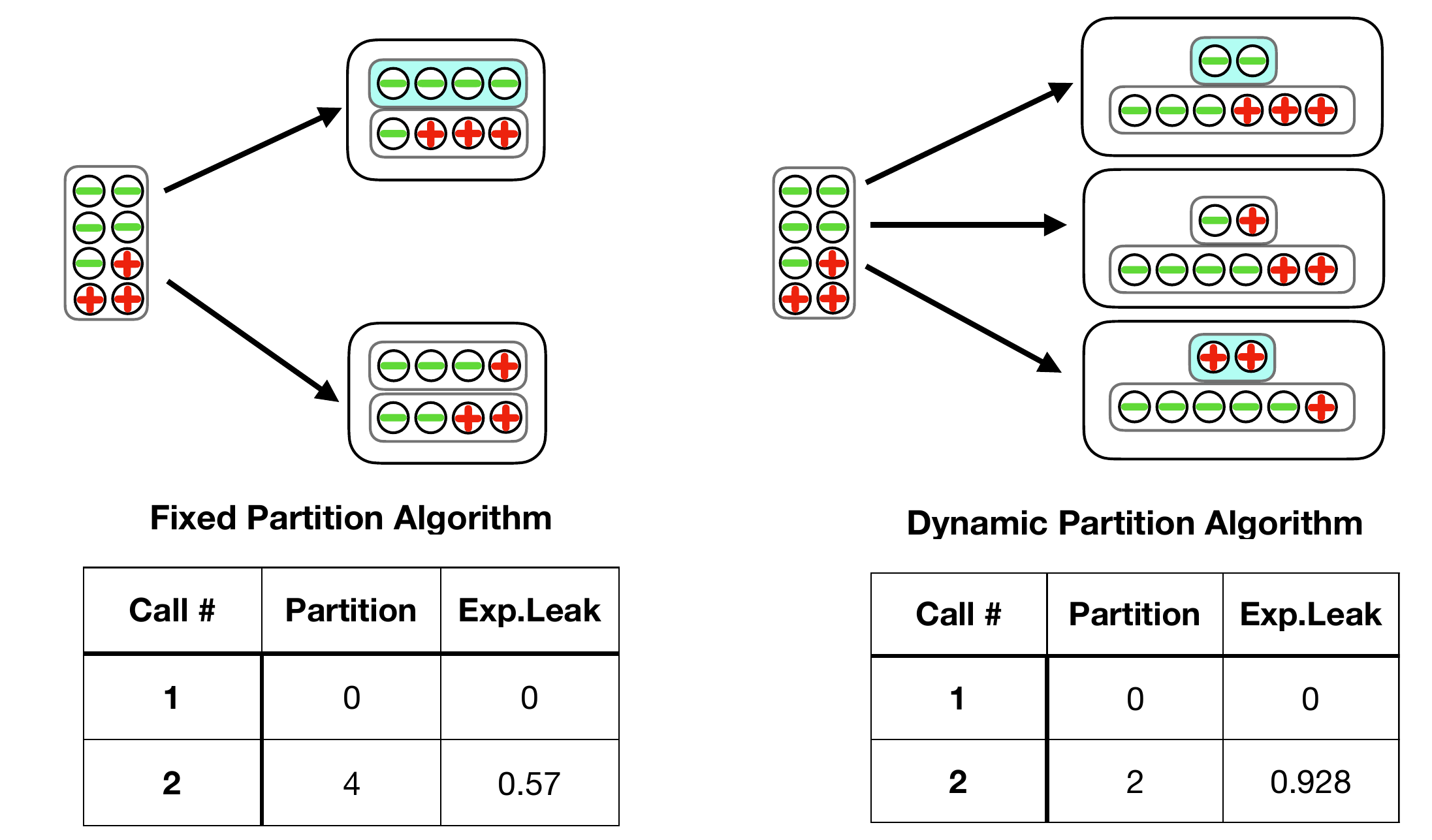} 
\label{fig:exp} } 
\caption{  {A comparison of different tree partition principles in hierarchical attacks on PSI-CA protocols.}} 
\label{exmp} 
\end{small}
\end{figure}
One possible direction to optimize the hierarchical structural attack is to adopt an uneven tree partition principle rather than divide all nodes evenly. We provide more insights through the following example.

Suppose an attacker possesses a dataset containing eight individuals, three of whom also are positive members and the rest of them are negative members. The adversary can request a PSI call twice. Our exploration focuses on two distinct attack strategies: fixed partitioning and dynamic partitioning.

In the fixed partition strategy, the adversary consistently bisects the tree, employing one of the child nodes for the next round. As shown in the figure, the attacker employing the fixed partition strategy initially executes a PSI run, identifying three positive members in their input but failing to deduce their identities. Then, the attacker randomly selects four individuals from the $\mathbb X$ for the second PSI run. As depicted in the figure, two potential scenarios arise:
The first scenario, with a probability of $\binom{5}{4}/\binom{8}{4}$, allows the attacker to discern the identities of four individuals.
The second scenario unfolds with a probability of $\binom{5}{3}\binom{3}{1}/\binom{8}{4}$, in which the attacker cannot infer any identities.
Following two PSI calls, the expected number of inferred memberships is $0.57$.
In contrast, the dynamic partitioning strategy involves an algorithm where the attacker can unevenly split the current branch, which then serves as the input for the next PSI call. The rationale is that the attacker, understanding the constraints of their protocol run time, opts to thoroughly identify a subset of individuals instead of selecting a larger quantity and failing to determine all their memberships. Case 2 from Figure \ref{exmp} illustrates the attack process using the dynamic partition strategy.

In this process, the attacker initiates a PSI run, identifies three individuals as positive members, and then randomly chooses two individuals for the second round. This can result in three potential outcomes:
Both selected individuals are negative members. This scenario, with a probability of $\binom{5}{2}/\binom{8}{2}$, results in a leakage of $2$.
The two individuals include one positive and one negative member, yielding no leakage.
Both selected individuals are negative. This outcome, with a probability of $\binom{3}{2}/\binom{8}{2}$, results in a leakage of $2$.
The expected leakage after two PSI calls is $0.928$, which exceeds that of the fixed partition scenario.

This raises a pertinent question: Given that the dynamic partition algorithm can potentially enhance inference efficiency, how can the adversary determine the optimal partition factor (for example, 2 in the example above)? In the following section, we introduce DyPathBlazer, a dynamic programming solution designed to resolve this optimal partition factor issue.



\subsection{DyPathBlazer: A bottom-up dynamic programming solution }\label{sec:4.1}

The proposed algorithm is depicted in Alg.\ref{alg:cap1} and can be summarized as follows. The attacker requests a PSI call for $\mathbb X$ and observes a return $C_{\mathbb X}$, the PSI call budget is also updated to $\tau -1$. These together determine an optimal partition factor $K$ (the determining function will be introduced later). He then randomly divides the tree into two child nodes. One with $K$ elements (denoted as left child), one with $|\mathbb X| -K$ elements (denoted as right child). He submits a child with a shorter length to the PSI protocol, the left child for example, and observes an outcome $C_{L}$. Then $C_{R} = C_{\mathbb X} - C_{L}$. The next step is to compare the expected leakage of the left child and the right child. The child node with a larger expected leakage is treated as the parent node and waiting for $K$ in the next iteration. The child node with smaller expected leakage is pushed into a priority queue. If all elements' membership in the current node is determined (all positive or all negative), the algorithm makes a prediction and dequeues a node with the highest priority as the current node. This process goes on until the protocol call budget is exhausted or all elements' membership in $\mathbb X$ are determined.

The selection of $K$ is dependent on three parameters: the number of elements in the parent node $|\mathbb N|$, the number of positive members of the elements $C_{\mathbb N}$, and the protocol run budget $\tau$. We define the tuple containing these three factors as a \textit{state}: $(|\mathbb N|,C_{\mathbb N},\tau)$.

\begin{algorithm}
\caption{DyPathBlazer}\label{alg:cap1}
\hspace*{\algorithmicindent} \textbf{Input:} A set $\mathbb X$ of target elements, pre-calculated Intermediate results $\Theta$, protocol invocation times $\tau$\\
 \hspace*{\algorithmicindent} \textbf{Output:} Classified sets $\mathbb{Z}_{pos}$, $\mathbb{Z}_{neg}$.
\begin{algorithmic}
\item Initialize $\mathbb Z_{pos} = \varnothing$, $\mathbb Z_{neg} = \varnothing$
\item $C_{\mathbb{X}} \gets$ PSI-CA$(\mathbb{X},\mathbb{Y})$, $\tau = \tau -1$.
\item \textsf{queue} $\gets$ $\{(|\mathbb X|, C_{\mathbb{X}}, \mathbb{X})\}$
\While{\textsf{queue} is not empty}
\State $(C_{\mathbb{N}}, \mathbb{N}) \gets$ \textsf{queue}.dequeue
\While{$0<C_{\mathbb{N}}<|\mathbb{N}|$ and $\tau>0$}
\State $K$ = $\Theta(\mathbb{N}, C_{\mathbb{N}}, \tau)$
\State $\mathbb{N}_K^L\gets \mathbb{N}[:K]$, $\mathbb{N}_K^R\gets \mathbb{N}[K:]$
\State $C_L\gets$ PSI-CA$(\mathbb{N}_K^L, \mathbb{Y})$, $\tau \gets \tau - 1$
\State $C_R \gets C_{\mathbb{N}} - C_L$
\If {$\Gamma^L_K\ge \Gamma^R_K$}
\State $(C_{\mathbb {N}}, \mathbb{N})$ $\gets(C_L,\mathbb {N}^L_K)$
\State \textsf{queue}.enqueue $\{(\Gamma^R_K, C_R, \mathbb{N}^R_K)\}$
\Else
\State $(C_{\mathbb {N}}, \mathbb{N})$ $\gets (C_R,\mathbb {N}^R_K)$
\State \textsf{queue}.enqueue $\{(\Gamma^L_K, C_L, \mathbb{N}^L_K)\}$
\EndIf
\EndWhile
\If{$C_{\mathbb{N}} = |\mathbb{N}|$} $\mathbb{Z}_{pos}\gets \mathbb{Z}_{pos}\cup \mathbb{N}$
\ElsIf{$C_{\mathbb{N}} = 0$} $\mathbb{Z}_{neg}\gets \mathbb{Z}_{neg}\cup \mathbb{N}$
\EndIf
\EndWhile\\
\Return $\mathbb Z_{pos}$, $\mathbb Z_{neg}$
\end{algorithmic}
\end{algorithm}

\begin{algorithm}
\caption{Memo generating function}\label{alg:memo}
 \hspace*{\algorithmicindent} \textbf{Output:} Memorized space $\Gamma$ and $\Phi$.
\begin{algorithmic}
\item Initialize $\Gamma$ and $\Phi$ as hash tables, where the key of $\Gamma$ is $(\mathbb N,C_{\mathbb{N}},\tau)$, and the key of $\Phi$ is $(\mathbb{N},C_{\mathbb{N}})$
\item $\Gamma$ (1,1,0) = 1; $\Gamma$ (1,0,0)=1;
\item $\Phi$ (1,1) = 0; $\Phi$ (1,0) = 1;
\State \Return $\Gamma$, $\Phi$
\end{algorithmic}
\end{algorithm}


\begin{figure}[t]
\begin{small}
\centering 
{ \includegraphics[width=0.4\textwidth]{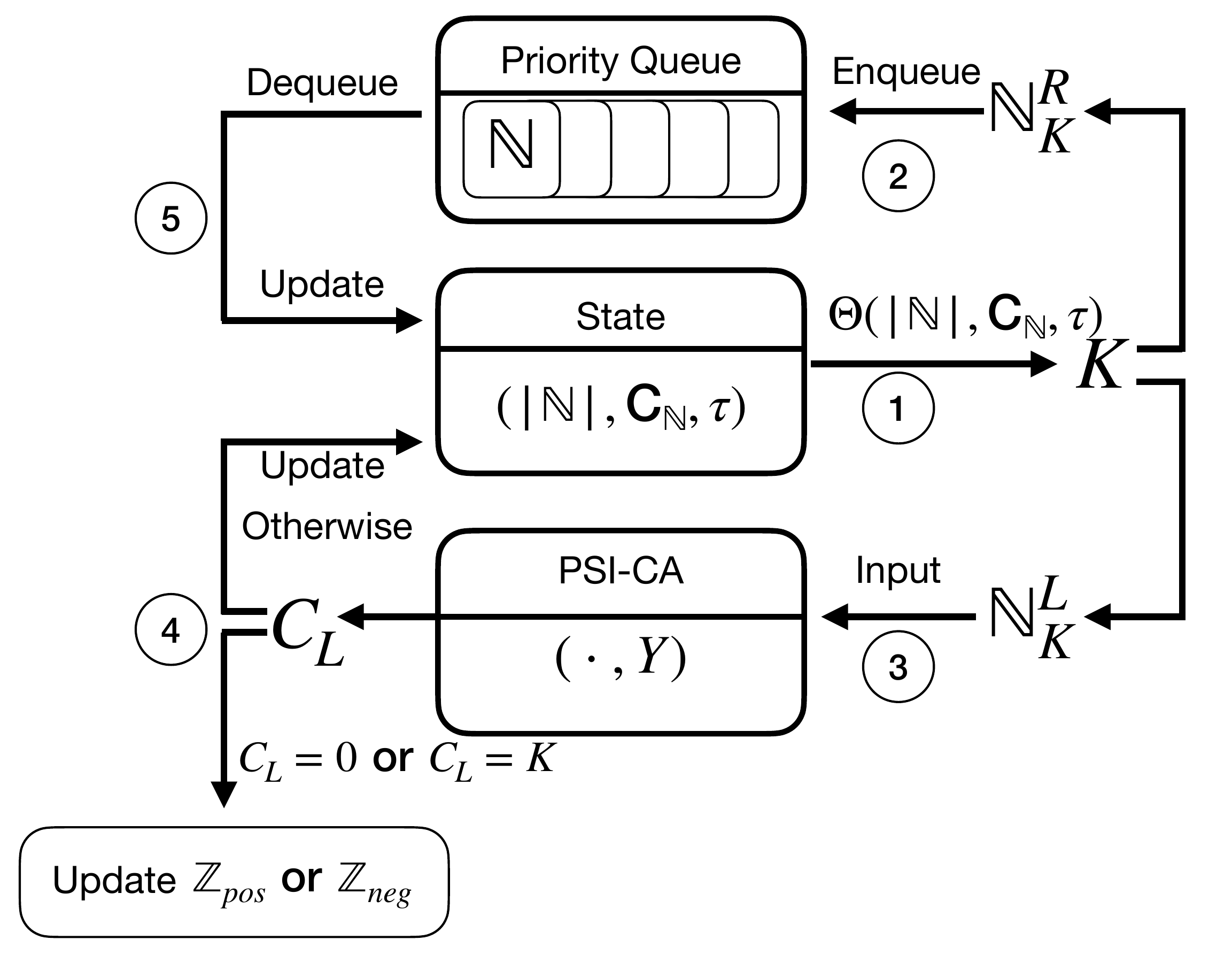} } 
\caption{DyPathBlazer model description } 
\label{det_att_model} 
\end{small}
\end{figure}

\begin{algorithm}
\caption{Optimal Tree partition algorithm based on dynamic programming}\label{alg:cap5}
\hspace*{\algorithmicindent} \textbf{Input:} Size of the victim set $N$, memberships in the victim set $m$, number of protocol innovation budget $\tau$\\
 \hspace*{\algorithmicindent} \textbf{Output:} $\Theta$, the optimal set partition principle.
\begin{algorithmic}
\State \Return $\Theta(\mathbb N,C_{\mathbb{N}},\tau)$ if exists.
\item Initiate $\Gamma(\mathbb N,C_{\mathbb{N}},\tau) = 0$, $\Theta(\mathbb N,C_{\mathbb{N}},\tau)=0$, $\Phi(\mathbb N,C_{\mathbb{N}})=N$
\For{1 $\le k \le (\mathbb{N}+1)//2$}
\State Calculate $\Gamma^{L}_k(\mathbb N,C_{\mathbb{N}},\tau)$ with \eqref{eq:left}
\State Calculate $\Gamma^{R}_k(\mathbb N,C_{\mathbb{N}},\tau)$ with \eqref{eq:right}
\State $\Gamma(\mathbb N,C_{\mathbb{N}},\tau-1)$ with \eqref{eq:findgm}
\EndFor
\State $\Theta(\mathbb N,C_{\mathbb{N}},\tau)=\text{argmax}_{k}\Gamma^k(\mathbb N,C_{\mathbb{N}},\tau)$
\If {$\Gamma(\mathbb N,C_{\mathbb{N}},\tau)\ge{\mathbb{N}}$}
\State $\Phi(\mathbb N,C_{\mathbb{N}})=\min\{\Phi(\mathbb N,C_{\mathbb{N}}),\tau\}$
\EndIf\\
\Return $\Theta(\mathbb N,C_{\mathbb{N}},\tau)$
\end{algorithmic}
\end{algorithm}

Note that different values of the partition factor $K$ may lead to different membership inference scenarios. Denote $\Gamma^k(|\mathbb N|,C_{\mathbb N},\tau)$ as the expected membership leakage when the partition factor is  $k$, given the current state.  Denote $\Theta(|\mathbb N|,C_{\mathbb N},\tau)=K$ as the partition rule, then:
\begin{equation}\label{eq:findk}
    K = \text{argmax}_{k}\Gamma_k(|\mathbb N|,C_{\mathbb N},\tau).
\end{equation}
We remove the subscription to denote the best-case expected leakage: $\Gamma(|\mathbb N|,C_{\mathbb N},\tau)=\Gamma_K(|\mathbb N|,C_{\mathbb N},\tau)$. In deriving the values of $\Gamma$, we also need to introduce another term $\Phi(|\mathbb N|,C_{\mathbb N})$, which denotes the smallest expected protocol run time the adversary needs to infer $\mathbb N$ target users' identity given $C_{\mathbb{N}}$ of them are positive memberships:
\begin{equation}
    \Phi(|\mathbb N|,C_{\mathbb N}) = \min_{\Gamma(|\mathbb N|,C_{\mathbb N},\tau) \ge N }{\tau}.
\end{equation}
After the partition, the current input set is divided into two parts with $K$ and $|\mathbb{N}|-K$ elements respectively. Denote $\mathbb {N}_K^L$ as the first set with $K$ elements, and $\mathbb {N}_K^R$ as the set with the other $|\mathbb{N}|-K$ elements.
Then the positive members in the first $K$ elements, PSI-CA result of $\mathbb{N}_K^L$ and $\mathbb Y$ can be denoted as a random variable $C_{L}$, with a probability distribution of:
\begin{equation}\label{eq:dist}
    \text{Pr}(C_{L}=c) = \frac{\binom{C_{\mathbb{N}}}{c}\binom{|\mathbb{N}-C_{\mathbb{N}}|}{K-c}}{\binom{\mathbb{N}}{K}}.
\end{equation}
We next show that $\Gamma_K(\mathbb N,C_{\mathbb {N}},\tau)$ can be derived by recursive functions:
After partition, the total leakage by choosing the left child becomes:
\begin{equation}\label{eq:left}
\begin{aligned}
    &\Gamma^L_K(\mathbb N,C_{\mathbb {N}},\tau) =\\ &\begin{cases}
E_{C_L}[\Gamma(K,C_{L},\tau-1)], \qquad\quad\qquad~~~~\text{if}~~ \Phi(K,C_{L}) \ge {\tau - 1},\\
K+E_{C_L}[\Gamma(|\mathbb N|-K,C_{\mathbb{N}}-C_{L},\tau-\Phi(K,C_{L}))], ~~ \text{otherwise}.
    \end{cases}
\end{aligned}
\end{equation}
Similarly,  the total leakage by choosing the right child becomes:
\begin{equation}\label{eq:right}
\begin{aligned}
    &\Gamma^R_K(\mathbb{N},C_{\mathbb{N}},\tau) =\\ &\begin{cases}
E_{C_L}[\Gamma(|\mathbb{N}|-K,C_{\mathbb{N}}-C_L,\tau-1)], \\~~~~~~~~~~~~~~~~~~~~~~~~~~~~\text{if}~~ \Phi(|\mathbb{N}|-K,C_{\mathbb{N}}-C_{L})\ge {\tau - 1},\\
|\mathbb{N}|-K+E_{C_L}[\Gamma(K,C_L,\tau-\Phi(|\mathbb{N}|-K,C_{\mathbb{N}}-C_{L}))], \\~~~~~~~~~~~~~~~~~~~~~~~~~~~~~~~~~~~~~~~~~~~~~~~~~~~~~~~~~~\text{otherwise}.
    \end{cases}
\end{aligned}
\end{equation}
Then,  
\begin{equation}\label{eq:findgm}
    \Gamma_K(|\mathbb{N}|,C_{\mathbb{N}},\tau)=\max\left\{\Gamma_K^L(|\mathbb{N}|,C_{\mathbb{N}},\tau),\Gamma_K^R(|\mathbb{N}|,C_{\mathbb{N}},\tau)\right\}.
\end{equation}
Intuitively, the expected membership leakage given of the current \textit{state} depends on the partition factor $K$. For a given $k$, the algorithm chooses from the left or the right child according to their corresponding expected leakage values: $\Gamma^L_K(|\mathbb{N}|,C_{\mathbb{N}},\tau)$ and $\Gamma^R_K(|\mathbb{N}|,C_{\mathbb{N}},\tau)$. The algorithm subsequently selects the branch with the maximum expected leakage value as the new input of the PSI, relegating the other part to the priority queue.
The processing framework is shown in Fig.\ref{det_att_model}.

Now suppose the left branch has a larger expected leakage than the right. After dividing $\mathbb{N}$ into $\mathbb{N}_K^L$ and $\mathbb{N}_K^R$, the algorithm calls the PSI-CA protocol and receives an observation of the cardinality of the intersection $C_L$. If $C_L=0$, none of the elements in the input $\mathbb{N}_K^L$ belongs to the intersection and $\mathbb{Z}_{neg}=\mathbb{Z}_{neg}\cup\mathbb{N}_K^L$. On the contrary, if $C_L = K$, all elements belong to the intersection and $\mathbb{Z}_{pos}=\mathbb{Z}_{pos}\cup\mathbb{N}_K^L$. Other than these two cases, the new $C_L$, together with $K$ and $\tau-1$ form a new state of the attack $(K,C_L,\tau-1)$. 

It is worth noting that once the expected membership leakage for the parent node is calculated, all the expected membership leakage of its child nodes are available, as they are the intermediate steps in calculating the parent node. i.e., to calculate $\Gamma(\mathbb{N},C_{\mathbb{N}},\tau)$, all $\Gamma(\mathbb{N}',C'_{\mathbb{N}},\tau')$ are required, where $1\le \mathbb{N}'\le\mathbb{N}$, $1\le C'_{\mathbb{N}}\le C_{\mathbb{N}}$ and $0<\tau'\le \tau$. The same rule also applies to the structure of $\Phi (\mathbb{N},C_{\mathbb{N}})$. When the algorithm is proceeding, any observations on the real number of positive memberships in the current node are included. Therefore, the policy space of $\Gamma$ only needs to be calculated once offline and can be directly used for further retrieval. This backtracking structure of building the memorization table of  $\Phi (\mathbb{N},C_{\mathbb{N}})$ in the dynamic programming solution is depicted in Fig. \ref{DP_backtrack}.

The following Theorem states the optimality of the proposed algorithm under the hierarchical attack structure.

\begin{thm}\label{thm:1}
    The proposed algorithm is optimal in maximizing the number of expected membership inferences under the hierarchical structure.
\end{thm}
\begin{proof}
The proof of Theorem \ref{thm:1} follows a simple concept containing two steps, in the first step, we show that under a given dividing value, the state of the membership inference attack can be decomposed into an average (expectation) of all possible sub-branches. Then in the second step, we show our algorithm always selects the optimal dividing value according to the maximum number of expectations of the inference.


\end{proof}

\begin{figure}[t]
\begin{small}
\centering 
{ \includegraphics[width=0.35\textwidth]{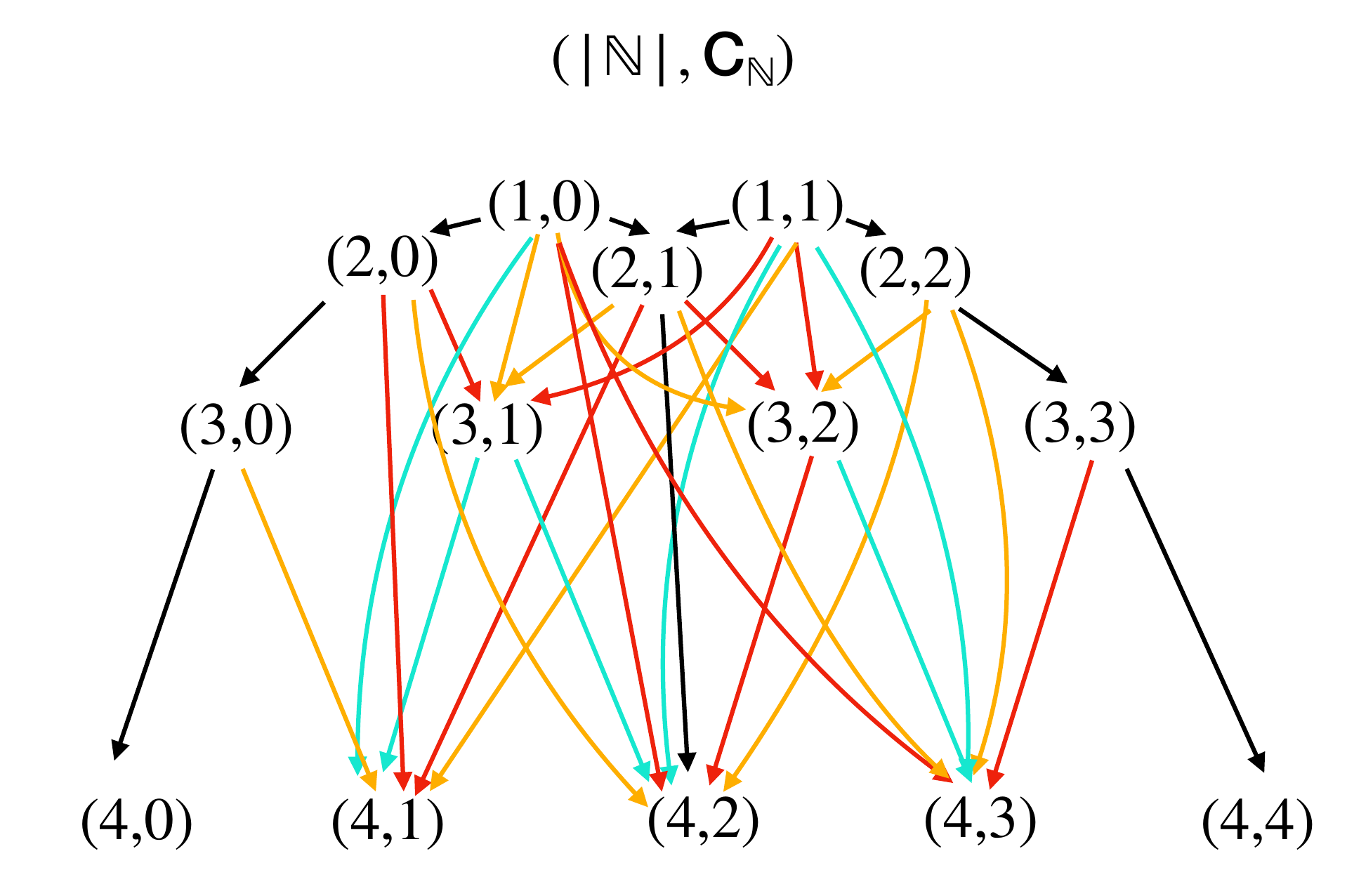} 
\caption{The backtracking structure of memorization table in the dynamic programming solution} 
\label{DP_backtrack}} 
\end{small}
\end{figure}

\textbf{Leakage Lowerbound}
We next analyze the worst-case performance of our proposed attack algorithm. For a given dataset, after the initial PSI call, the state variables of $\Gamma$ are available. Denote them as $(|\mathbb{N}|,C_{\mathbb{N}},\tau)$. With $\Theta(\mathbb{N},C_{\mathbb{N}},\tau)$, the partition factor $K$ can be determined. The uncertainty of the attack performance comes from the randomness in picking the $K$ out of the $|\mathbb{N}|$ individuals for the following protocol call. 
To derive the lower bound of the performance, here we analyze the worst-case scenario of selection. Then the left branch terminates if $K-C_L=0$ or $C_L=0$, and the right branch terminates if $C_{\mathbb{N}}-C_L=0$ or $\mathbb{N}-K-C_{\mathbb{N}}+C_L=0$. Note that if any of these conditions approach to $0$, the corresponding branch tends to terminate. The worst-case value of $C_L$ satisfies the following condition:
\begin{equation}\label{eq:wcm}
    \max_{C_L}[(K-C_L)^2+C_L^2+(C_{\mathbb{N}}-C_L)^2+(|\mathbb{N}|-K-C_{\mathbb{N}}+C_L)^2].
\end{equation}

\begin{prop}
The worst-case sampled $C_L$ for the partition factor $K$ is $C_L=C_{\mathbb N}K/\mathbb{N}$.
\end{prop}

Then the lower bound can be numerically derived using the dynamic programming algorithm in DyPathBlazer with the following modification: at each node, instead of asking for a PSI run, the adversary calculates $C_L$ according to \eqref{eq:wcm}. Therefore, the online attack algorithm can be transferred to a local attack except for the initial protocol run. 

\begin{rmk}
Compared with the leakage lower-bound of the attack proposed in  \cite{281334}, it is straightforward $K = |\mathbb N|/2$, and the worst case $C_L = C_{\mathbb{N}}/2$. 
\end{rmk}



We theoretically compare the lower bound of the membership leakage of the DyPathBlazer and the algorithm in \cite{281334} (denoted as USENIX 22) under two cases: case 1) among $100$ individuals, $50$ are positive members. case 2) among $100$ individuals, $20$ are postiive members. The comparison results are plotted in Fig. \ref{fig:lowerbound_com}.
\begin{figure}[t]
\begin{small}
\centering 
{ \includegraphics[width=0.23\textwidth]{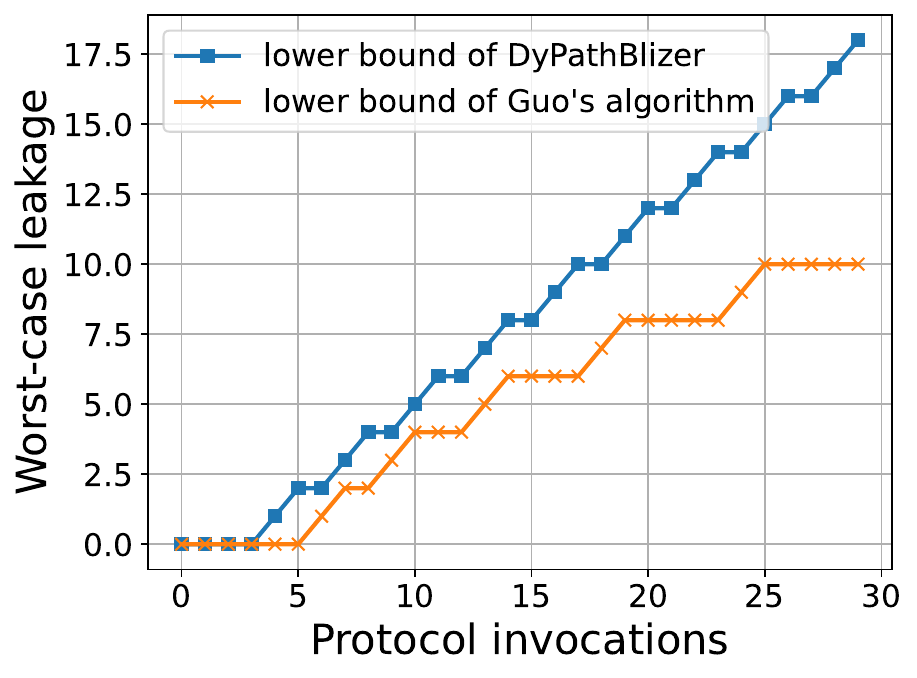} 
\label{fig:lower_1} } 
{ \includegraphics[width=0.23\textwidth]{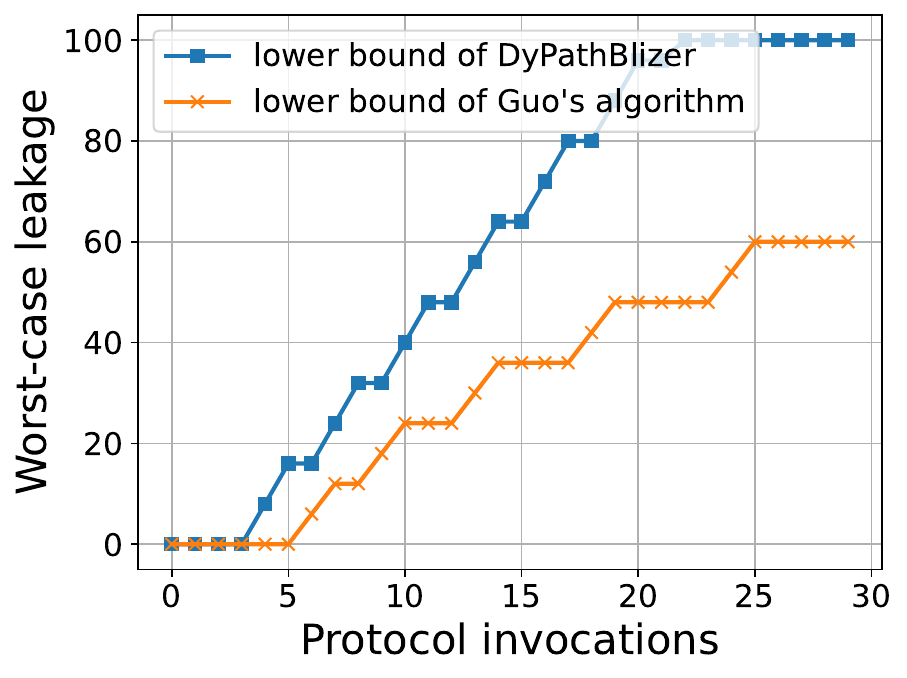}
\label{fig:lower_2} } 
\caption{Performance lower bound comparison (worst-case leakage comparison) of DyPathBlazer and  {Guo's algorithm} in \cite{281334}. Both cases consider a dataset of $100$ individuals. Case 1)  assumes $50$ positive members, case 2) assumes $10$ positive members.  {Higher lower bounds from DyPathBlazer guarantees better efficiency in the worst-case scenario.
DyPathBlazer's lower bound surpasses that of Guo's algorithm lower bound, ensuring improved efficiency in the worst-case scenario.}} 
\label{fig:lowerbound_com} 
\end{small}
\end{figure}

\begin{figure*}[t]
\begin{small}
\centering 
{ \includegraphics[width=0.7\textwidth]{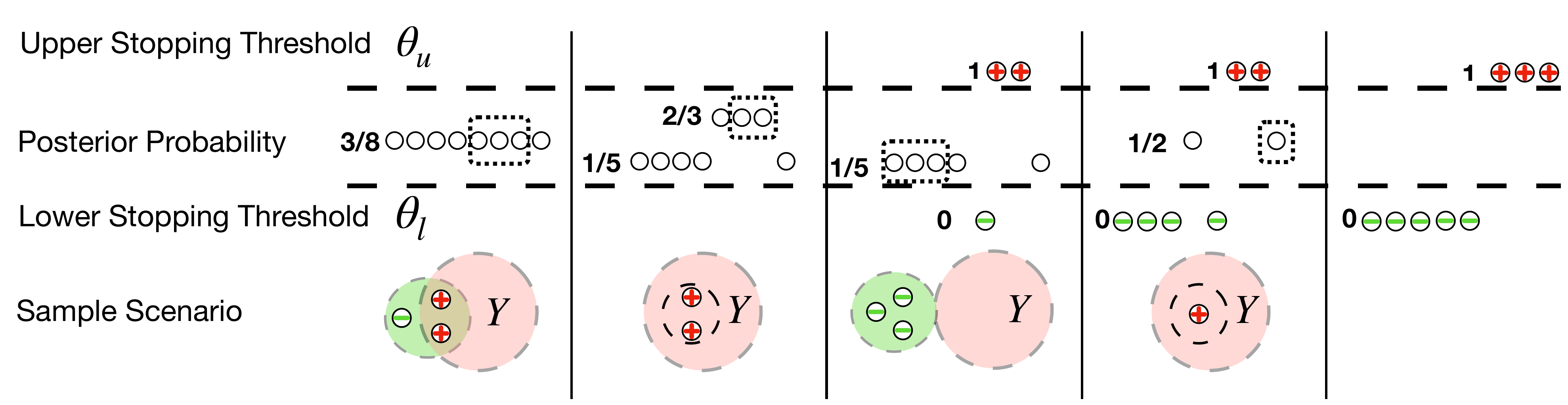} 
} 
\caption{An illustration of statistical membership inference attack.} 
\label{sta_att} 
\end{small}
\end{figure*}

\subsection{TreeSumExplorer: Deterministic attack on PSI-SUM}\label{sec:4-c}

In this part, we propose an attack algorithm to evaluate the privacy leakage of the PSI-SUM protocol. This protocol is a black box that receives
a set $\mathbb X$ from one party and a table $\mathbb Y ={(y_i
, v_i)}$ of index-value
pairs from the other party. It internally aggregates the values
associated with the indices in the intersection $\mathbb X \cap \mathbb Y$ and gets
the sum $\sum_{y_i\in \mathbb X\cap \mathbb Y} v_i$
. Then, the protocol sends this sum and the
intersection size $|\mathbb X\cap \mathbb Y|$ to the party that inputs the table.

Intuitively, the leakage of the PSI-SUM should be greater than the leakage of PSI-CA, as more information (summation of  the intersection) is observable during each protocol run. We next summarize the steps to take advantage of this information and conduct a more efficient privacy evaluation.

The basic idea behind this approach is to offline search all possible combinations that match the summation result. For example, after a protocol is called, the adversary observes the following: input size $|\mathbb N|$, returned number of members $C_{\mathbb{N}}$, and the summation of the intersected members' value \textsf{SUM}. Instead of directly applying the cardinality-based attack, the adversary is capable of traversal all possible combinations of the $C_{\mathbb{N}}$ out of $|\mathbb N|$ users in the dataset, and checking if the summation of their data values matches \textsf{SUM}. So the membership identity of these $|\mathbb N|$ users can be inferred in one shot (best case depending on the number of possible combinations). Such a problem is defined as an N-SUM problem (summation of  {$C_{\mathbb{N}}$} elements in the array,  {denoted as $\textsf{Arr}$}, matches a \textsf{SUM}) \cite{SOMA200257}, and a common algorithm for solving the N-SUM problem is summarized as follows:

\begin{algorithm}
\caption{Privacy evaluation algorithm for PSI-SUM}\label{alg:cap7}
\hspace*{\algorithmicindent} \textbf{Input:} A set $\mathbb X$ of target elements, protocol invocation times $\tau$, maximum set size $N$.\\
 \hspace*{\algorithmicindent} \textbf{Output:} Predicted sets: $\mathbb Z_{pos}, \mathbb Z_{neg}$.
\begin{algorithmic}
\item Initialize $\mathbb Z_{pos} = \varnothing$, $\mathbb Z_{neg} = \varnothing$
\item Find a subset $\mathbb N$ according to the computation power
\While{$0<\tau$}
\State $(C_{\mathbb N}, \textsf{SUM})\gets$ PSI-SUM($\mathbb N, \mathbb Y$),  $\tau \gets \tau -1$
\State \textsf{Arr} = N-SUM($\mathbb N, \textsf{SUM}, C_{\mathbb N}$)
\State \textsf{Priority} $\gets |\mathbb N|/|\textsf{Arr}|$. 
\State \textsf{queue}.enqueue $\{\textsf{Priority}, C_{\mathbb N}, \textsf{SUM}, \mathbb N, \textsf{Arr})\}$
\While{\textsf{queue} is not empty}
\State ($C_{\mathbb N}, \textsf{SUM}, \mathbb N, \textsf{Arr}$) $\gets$ \textsf{queue}.dequeue()
\While{$1<|\textsf{Arr}|$, and $0<\tau$}
\State $\mathbb N_L$ = Arr[0], $\mathbb N_R$ = $\mathbb N \setminus \mathbb{N}_L$
\State $(C_L,\textsf{SUM}_{L}) \gets$ PSI-SUM($\mathbb{N}_L, \mathbb Y$), $\tau \gets \tau - 1$
\State  $C_R \gets C_{\mathbb N} - C_L$, $\textsf{SUM}_R\gets \textsf{SUM}-\textsf{SUM}_L$
\State $\textsf{Arr}_L$ = N-SUM($\mathbb{N}_L, \textsf{SUM}_L, C_L$)
\State $\textsf{Arr}_R$ = N-SUM($\mathbb{N}_R, \textsf{SUM}_R, C_R$)
\State Compute priorities according to \eqref{sum_pri}
\State Continue with the high-priority branch
\State Push low-priority branch to the queue
\EndWhile
\If{$|\textsf{Arr}|=1$} 
\State $\mathbb Z_{pos}\gets \mathbb Z_{pos}\cup \textsf{Arr}$, $\mathbb Z_{neg}\gets \mathbb Z_{neg}\cup \mathbb N/ \textsf{Arr}$
\EndIf
\EndWhile
\EndWhile\\
\Return $\mathbb Z_{pos}, \mathbb Z_{neg}$
\end{algorithmic}
\end{algorithm}

\begin{itemize}
    \item Sort the array:  {The $\textsf{Arr}$} is sorted in ascending order.
\item Check the Base Case: If  {$C_{\mathbb{N}}=2$}, the problem is reduced to a 2-SUM problem. This is solved by initializing two pointers at both ends of the array and moving them toward each other until they meet. If the sum of the numbers at the pointer indices equals the target, then record this pair. The pointers are adjusted based on whether the current sum is less than or greater than the target.
\item Handle the Recursive Case: If  {$C_{\mathbb{N}} > 2$}, treat the problem as an  {$(C_{\mathbb{N}}-1)$}-SUM problem for each element in the array. This involves iterating over the array, and for each element, recursively solving an  {$(C_{\mathbb{N}}-1)$}-SUM problem with a new target that is equal to the original target minus the current element.
\end{itemize}

The computation complexity of the N-SUM solution described above is $\mathcal{O}(|\mathbb {N}|^{C_{\mathbb{N}}-1})$, which depends on the length of the input and the number of positive membership individuals. There is also a line of work improving this complexity \cite{chen2023subset, 4568054, nsum2}, however, is out of the scope of this paper. 

There could be multiple combinations of users whose summation matches \textsf{SUM} returned. The possibilities form a candidate set for the input of the next iteration. The adversary, then, randomly picks one of these combinations $\mathbb {N}'$ as the input to call for another protocol run. If the returned $\textsf{SUM}'$ matches the previous $\textsf{SUM}$, the selected combination are all positive memberships. Further, all  memberships in $\mathbb {N}$ are determined. Otherwise, the adversary has selected an incorrect combination as input. The adversary observes and obtains two sub N-SUM problems: a. N-SUM($\mathbb N'$, \textsf{SUM}', $C_{\mathbb N}$) and b. N-SUM($\mathbb N \setminus \mathbb N'$, \textsf{SUM}-\textsf{SUM}', $C_{\mathbb{N}}-C_{\mathbb{N}'}$). For each subproblem, the priority  {the sub-problem a or b} is calculated as follows:
\begin{equation}\label{sum_pri}
    \textsf{Priority} = \frac{|\mathbb N|}{| {\textsf{Arr}|}}.
\end{equation}

The adversary selects $\max\{\textsf{Priority}_a,\textsf{Priority}_b\}$ and continues his inference with the branch with higher priority, then pushes the lower priority branch to the priority queue. The algorithm for TreeSumExplorer is summarized in the Alg. \ref{alg:cap7}.

\section{ActBaysian: Statistical Attack on PSI-CA}\label{sec:4.2}

Different from the deterministic attack, where the adversary takes each individual's membership as constant. In the context of a statistical attack, each individual's membership is regarded as a binary random variable. This allows the adversary to design a soft-stopping criterion to finish his guessing. Therefore, the attacker's goal is to select subsets of $\mathbb X$ which helps accurately infer the identity of individual members.  Denote $L_i$ as a binary random variable that $L_i =1$ if $x_i\in \mathbb{X}\cap \mathbb{Y}$. ($L_i=0$ otherwise). Then the memberships in $\mathbb X$ can be specified as a random vector $p$, where $p[i]$ denotes the probability that $\text{Pr}(L_i = 1)$. Suppose for PSI attack runtime $t \in[1,..,\tau]$, the attacker selects a subset $\mathbb S_t$ of $\mathbb X$ as the input of PSI, let $O_1,...,O_{\tau}$ denote the release from the PSI protocol, then $O_t = |\mathbb S_t\cap \mathbb Y|$. 

In this section, we propose a statistical attack algorithm based on Bayesian posterior update and adaptive learning. We call this algorithm ActBaysian algorithm. The active learning process guarantees the adversary selecting the most informative subset of $\mathbb S_t$ at each time $t$. The posterior update process enables the adversary to update his belief on $L_i$ based on the observations.

\subsubsection{Posterior belief update}
After each PSI-CA call, the adversary observes the cardinality of the positive members contained in the input dataset. Based on this observation, the adversary updates their posterior belief. The update rule follows a maximum likelihood criterion. For a given pair of PSI input/output, denoted as $\mathbb S_{t}$ and $O_t$ respectively, the updated prior for each individual in the subset becomes:
\begin{equation}\label{update1}
p[i] = \frac{O_t}{|\mathbb S_t|}.
\end{equation}

This update reflects the ratio of positive members observed ($O_t$) to the cardinality of the input dataset ($|\mathbb S_t|$). The higher the observed positive count relative to the size of the input subset, the higher the posterior belief assigned to each individual.

On the other hand, all sets containing $\mathbb S_t$ can also update their posterior beliefs. Let $\bar{\mathbb S}_t$ be the subset such that $\bar{\mathbb S}_t\cap \mathbb S_t=\varnothing$, and the result of PSI-CA$(\mathbb S_t\cup\bar{\mathbb S}_t)$ is known, denoted as $O_{tot}$. After observing the output $O_t$, the posterior belief of all individuals' membership in $\bar{\mathbb S}_t$ becomes:
\begin{equation}\label{update2}
p[i] = \frac{O_{tot}-O_{t}}{|\bar{\mathbb S}_t|}.
\end{equation}

This update is based on the remaining positive count ($O_{tot}-O_{t}$) after subtracting the observed positive count $O_t$ from the known total positive count $O_{tot}$. The posterior belief is calculated by dividing this remaining positive count by the cardinality of $\bar{\mathbb S}_t$.
By updating the prior and posterior beliefs based on the observed positive counts and known total positive counts, the adversary can refine their belief about the membership status of each individual in the respective subsets. This iterative belief update process enables the adversary to incorporate new information and adjust their inference based on the observed results.


\textbf{Stopping criterion:} In the context of a statistical attack, the adversary has the flexibility to enhance their inference power at the cost of sacrificing some inference accuracy. This is achieved through the use of a soft-stopping criterion, which involves setting an upper threshold $\theta_{u}$ and a lower threshold $\theta_{l}$. The adversary classifies an individual's membership as positive if their posterior probability $p[i]$ is greater than or equal to $\theta_{u}$, and as negative if $p[i]$ is less than or equal to $\theta_{l}$.
It is important to note that the selection of these thresholds plays a crucial role in determining the performance of the attack. The threshold values significantly influence the trade-off between inference accuracy and the number of individuals' memberships inferred: 
If the upper threshold $\theta_{u}$ is set close to 1 and the lower threshold $\theta_{l}$ is set close to 0, the inference attack achieves high accuracy. However, there may be a limited number of individual memberships inferred. This occurs because the adversary requires a higher level of certainty before classifying an individual as a positive or negative member. Consequently, only individuals with posterior probabilities close to the extremes will be confidently classified, potentially leading to a reduced number of identified individuals.
On the other hand, if both the upper threshold $\theta_{u}$ and the lower threshold $\theta_{l}$ are set close to 0.5, a larger proportion of individuals' memberships can be inferred. However, the inference accuracy may be compromised. With lower thresholds, even samples with relatively uncertain posterior probabilities will be classified, resulting in a higher number of individuals being identified. Nevertheless, the trade-off is that the accuracy of these classifications may be lower compared to using more stringent thresholds.
\subsubsection{Active learning for input set selection}
The active learning contains two phases: 1. constructing candidate input sets by comparing the absolute distance between each posterior and the threshold.  2. determine the input set based on the minimized Manhattan distance.

\textbf{Minimize absolute distance:} During the attack process, the adversary aims to select individuals who are likely to meet the stopping criterion as the input for the PSI protocol. This selection is done in two directions: individuals whose prior probabilities are close to the upper threshold $\theta_{u}$ and individuals whose priors are close to the lower threshold $\theta_{l}$. The absolute distance between each individual's posterior and the upper/lower threshold is defined as follows:
\begin{equation}
\begin{aligned}
     &d_i^u = |p[i] - \theta_u|\\
    &d_i^l = |p[i] - \theta_l|.  
\end{aligned}
\end{equation}
For a given iteration, the minimized distance in the posterior vector is defined as $d_{\min}^u = \min_i d_i^u$, and $d_{\min}^l = \min_i d_i^l$, respectively.
However, directly selecting individuals with the highest or lowest priors introduces no randomness to the output, potentially leading to a dead loop in the attack. To mitigate this issue, two factors are introduced: the sampling rate $r$ and the tolerance factor $tol$.

The sampling rate and tolerance factor together determine the grouping principles for selecting the input individuals. The candidate input for the next iterations includes:

\begin{equation}
\begin{aligned}
\mathbb S^{u} &= \bigcup_{d^u_i[i]-d_{\min}^u\le tol} x_i \cdot \mathbf{1}_{sample} \\
\mathbb S^{l} &= \bigcup_{d^l_i[i]-d_{\min}^l\le tol} x_i \cdot \mathbf{1}_{{sample}},
\end{aligned}
\end{equation}

where $\mathbf{1}_{{sample}}$ is an indicator function that is determined by the sampling result. Intuitively, a larger tolerance factor $tol$ results in grouping more individuals who are not prone to meet the stopping criterion. This allows for a broader exploration of individuals in the search for those whose posterior probabilities may cross the thresholds. However, a larger candidate size increases the number of individuals whose memberships can be inferred in one iteration.
On the other hand, a small sampling rate $r$ tends to narrow down the inference scope quickly. The attack can rapidly infer a smaller portion of individuals with positive or negative memberships. However, this approach may lead to decreased overall inference efficiency since the algorithm becomes prone to depth-first search (DFS) behavior.
The choice of the sampling rate and tolerance factor depends on the specific attack objectives and constraints. By carefully tuning these parameters, the adversary can balance the trade-off between the efficiency of the attack, the number of inferences made in each iteration, and the overall accuracy of the inference process. We present a detailed analysis in numerical evaluation.

\textbf{Minimize Manhattan distance:} 
As the adversary keeps calling the PSI protocol, each individual's posteriors are pushing either towards $\theta_u$ or $\theta_l$. Given a tolerant factor $tol$, before each PSI call, the adversary has the choice to select from $\mathbb S_u$ and $\mathbb S_l$ as his input for the next PSI-CA call. This can be achieved by comparing the minimized averaged Manhattan distance between $\mathbb S_u$ and $\theta_{u}$, $\mathbb S_l$ and $\theta_l$:
\begin{equation}\label{eq:dis}
\begin{aligned}
    D_u = &\sum_{i: x_i \in \mathbb S_u} d^u_i/|\mathbb S_u|\\
    D_l = &\sum_{i: x_i \in \mathbb S_l} d^l_i/|\mathbb S_l|.\\
\end{aligned}
\end{equation}
Then, the adversary selects the input with a smaller distance toward the threshold. 
The algorithm is summarized in Alg. \ref{alg:sta1}.

\begin{algorithm}
\caption{Statistical Membership Inference Attack}\label{alg:sta1}
\hspace*{\algorithmicindent} \textbf{Input:} Victim set $\mathbb X$, $\theta_{u}$ upper threshold, $\theta_{l}$ lower threshold, PSI call budget $\tau$, sampling rate $r$, tolerance factor $tol$.\\
 \hspace*{\algorithmicindent} \textbf{Output:} Prediction of $\mathbb Z_{pos}$ and $\mathbb Z_{neg}$
\begin{algorithmic}
\item Initialize prior distribution for each user $\{p[i]=0.5\}_{i=1}^N$ 
\While{$\tau >0$ and $|\mathbb Z_{pos}+\mathbb Z_{neg}|<|\mathbb{N}|$}
\State $\mathbb S_{u}\gets \bigcup_{\{d^u_i[i]-d_{\min}^u\le tol\}} x_i$, sampling with $r$
\State $\mathbb S_{l}\gets \bigcup_{\{d^l_i[i]-d_{\min}^l\le tol\}} x_i$, sampling with $r$
\State Calculate $D_u$, $D_l$ with \eqref{eq:dis}. 
\If{$D_u\le D_l$}
\State $\mathbb S \gets  \mathbb {S}_u$
\Else
\State $\mathbb S \gets \mathbb {S}_l$
\EndIf
\State $C_{\mathbb{S}} = $ PSI-CA $(\mathbb S, \mathbb Y)$
\State Update posteriors according to \eqref{update1} and \eqref{update2}
\For{$1\le i\le |\mathbb N|:$}
    \State $\mathbb Z_{pos} = \mathbb Z_{pos}\cup x_i$ if $p[i]<\theta_{l}$ 
    \State $\mathbb Z_{neg} = \mathbb Z_{neg}\cup x_i$ if $p[i]>\theta_{u}$ 
    \EndFor    
\EndWhile\\
\Return $\mathbb Z_{pos}$, $\mathbb Z_{neg}$
\end{algorithmic}
\end{algorithm}

\subsection{Error bound}
The statistical attack incurs classification error corresponding to the upper and the lower stopping threshold: A group of individuals is classified as positive members if most of them are positive members. Similarly, a group of individuals are classified as negative members if most of them are negative members. Specifically, we define Type I, Type II, and misclassification rate for an individual $x_i$ as follows:
\begin{equation}
\begin{aligned}
    P^e_{type I}= &\text{Pr}(x_i \in \mathbb X\cap \mathbb Y|x_i \in \mathbb Z_{neg}),\\
    P^e_{type II}= &\text{Pr}(x_i \notin \mathbb X\cap \mathbb Y|x_i \in \mathbb Z_{pos}),\\
     {P}^e_{mis} = & P^e_{type I}\text{Pr}(x_i \in \mathbb Z_{neg}) + P^e_{type II}\text{Pr}(x_i \in \mathbb Z_{pos}).
\end{aligned}
\end{equation}
When the stopping state containing $x_i$ is $(\mathbb{N}, C_{\mathbb N}, \tau)$, from algorithm \ref{alg:sta1},  the upper bound of these error probabilities are summarized in proposition \ref{prop:2}.
\begin{prop}\label{prop:2} The probability upper bounds of Type I, Type II error, and the misclassification rate are as follows: 
\begin{equation}
    \begin{aligned}
    &P^e_{type I}\le 1-C_{\mathbb N}/|\mathbb{N}| ~ \text{if} ~ \theta_u\le C_{\mathbb N}/|\mathbb{N}|\le 1, \\
    &P^e_{type II}\le C_{\mathbb N}/|\mathbb{N}| ~ \text{if} ~ 0\le C_{\mathbb N}/|\mathbb{N}|\le\theta_l,\\
    &P^e_{mis}\le 1-\theta_u +\theta_l.
    \end{aligned}
\end{equation}
\end{prop}
It is straightforward to extend the error probabilities to correct guessing probabilities: 
\begin{equation}
    \begin{aligned}
    &P_{TP} = 1- P^e_{type I} \ge C_{\mathbb N}/|\mathbb{N}| ~\text{if} ~ \theta_u\le m/N\le 1, \\
    &P_{FP} = 1 -P^e_{type II} \ge 1-C_{\mathbb N}/|\mathbb{N}| ~\text{if} ~ 0\le C_{\mathbb N}/|\mathbb{N}|\le\theta_l.\\
    \end{aligned}
\end{equation}
where $P_{TP}$ stands for true positive probability and $P_{FP}$ stands for false positive probability.

\begin{figure*}[htp]
\begin{small}
\centering 
{ \includegraphics[width=0.32\textwidth]{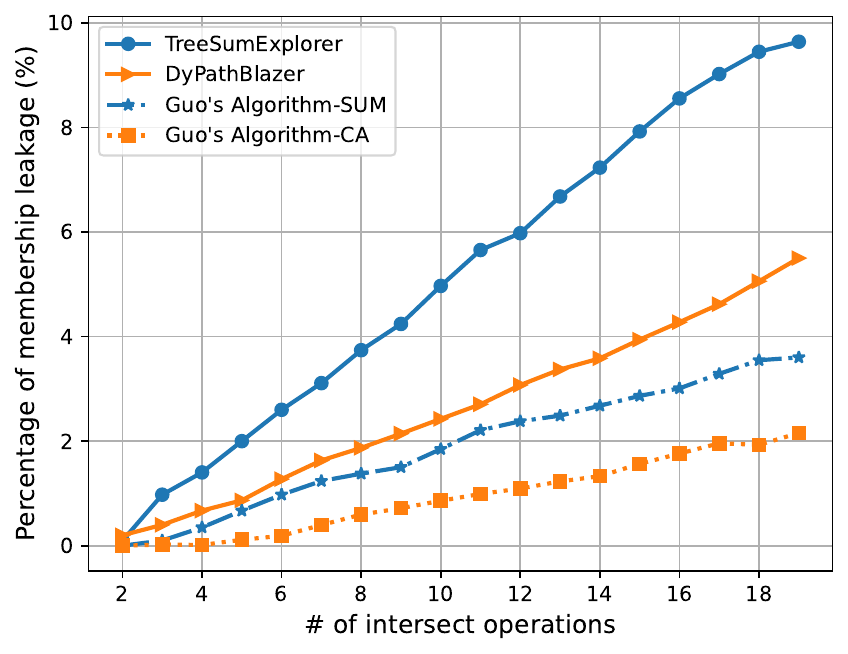}
\label{fig:tao_a_1} } 
{ \includegraphics[width=0.32\textwidth]{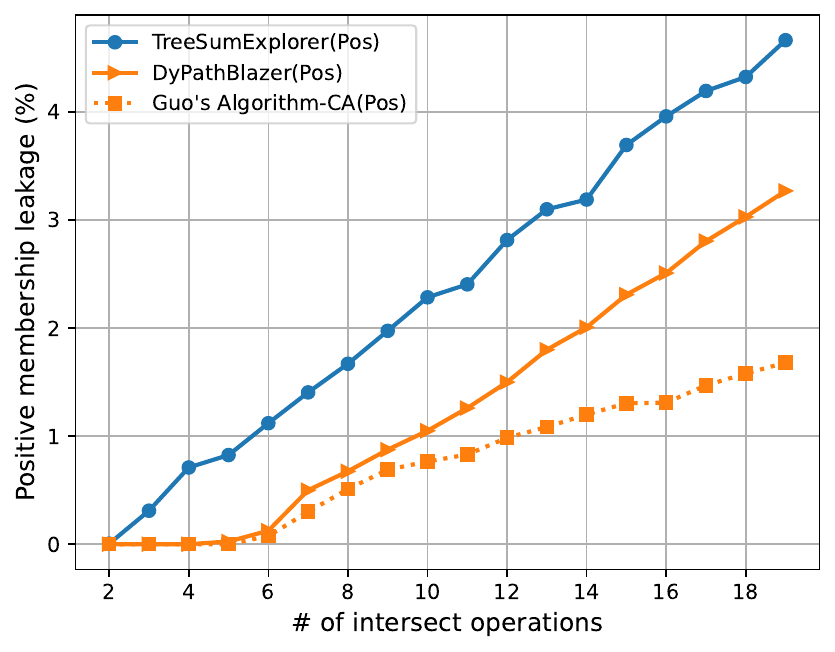} 
\label{fig:tao_a_2} } 
{ \includegraphics[width=0.32\textwidth]{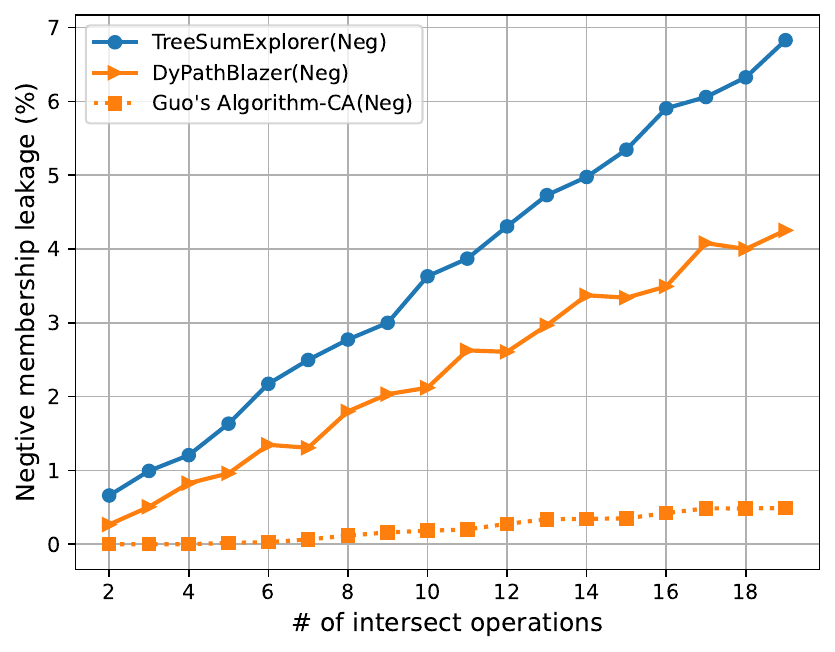}
\label{fig:tao_a_2} } 
{ \includegraphics[width=0.31\textwidth]{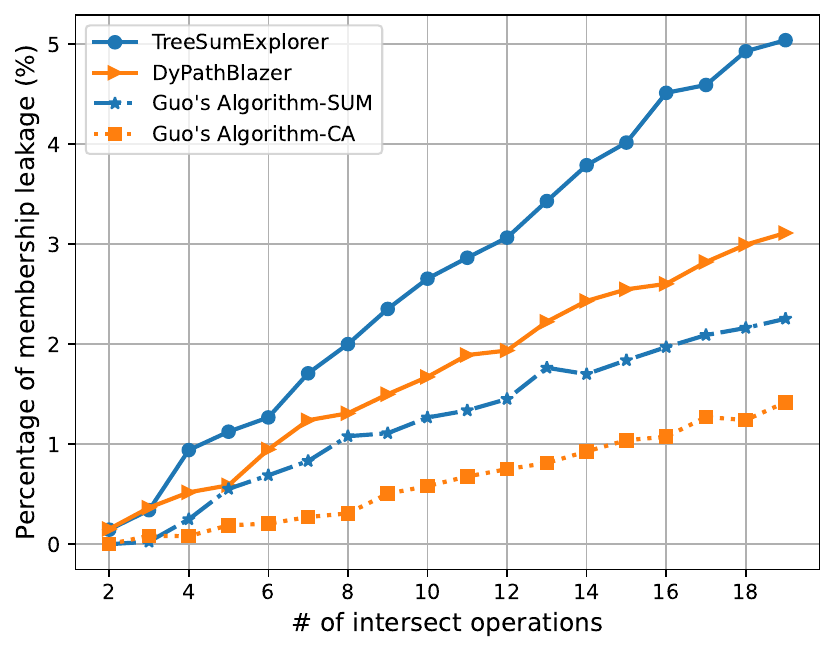}
\label{fig:ads_a_1} } 
{ \includegraphics[width=0.32\textwidth]{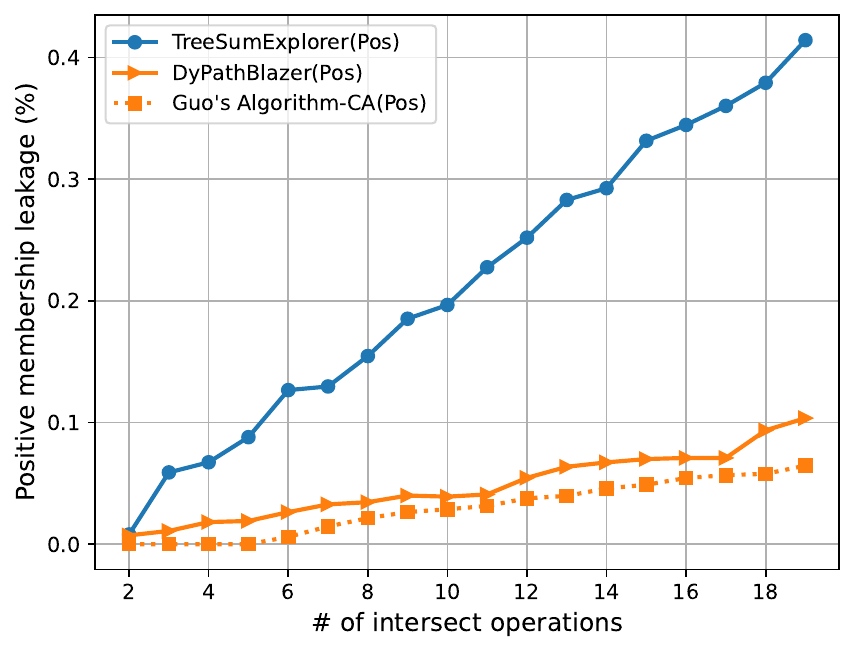} 
\label{fig:ads_a_2} } 
{ \includegraphics[width=0.32\textwidth]{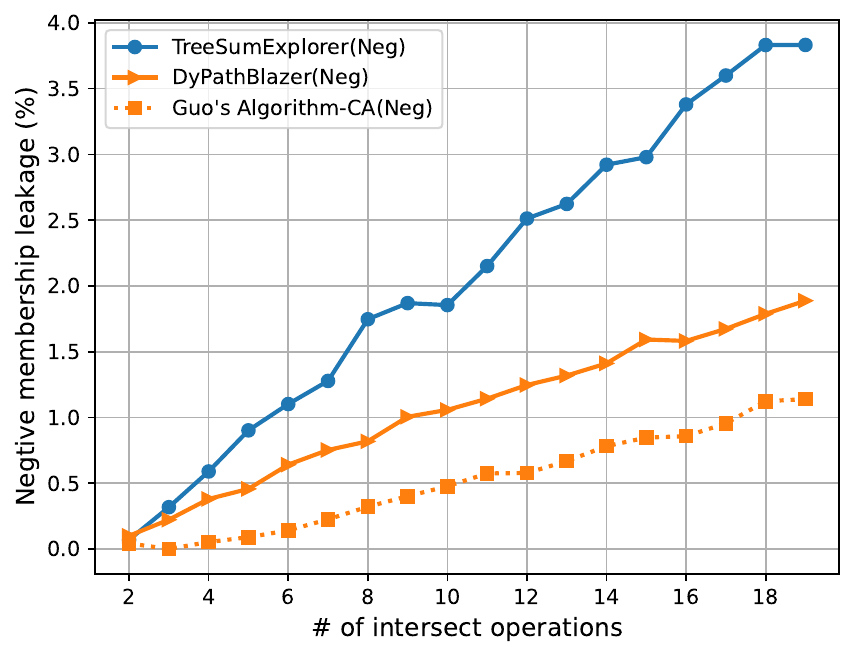}
\label{fig:ads_a_2} } 
\caption{Membership information leakages in PSI-SUM under different attacks.  {Column 1 displays the overall membership leakage, while Column 2 represents positive membership leakage, and Column 3 represents negative membership leakage. The first row corresponds to cases where the product company targets the advertising company, and the second row corresponds to cases where the advertising company targets the product company.}} 
\label{real_data1} 
\end{small}
\end{figure*}


\begin{figure*}[t]
\begin{small}
\centering 
{ \includegraphics[width=0.32\textwidth]{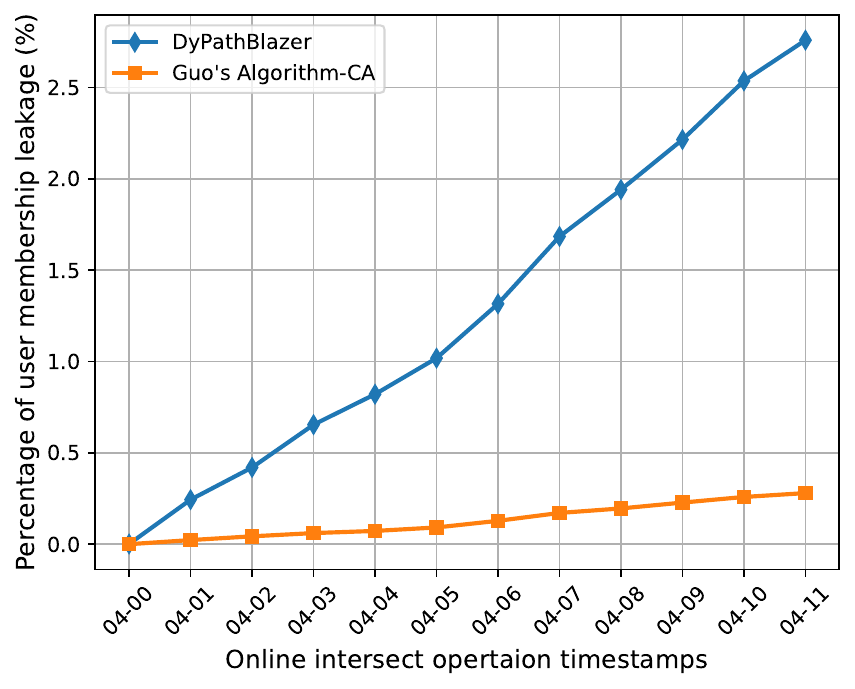}
\label{fig:covid1} } 
{ \includegraphics[width=0.32\textwidth]{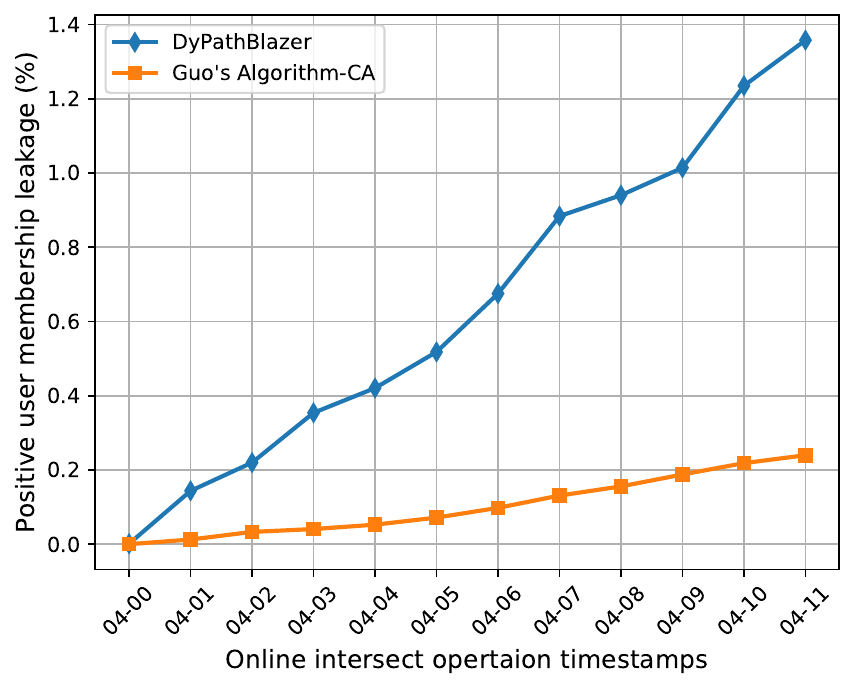} 
\label{fig:covid2} } 
\caption{Membership information leakages in PSI-CA under different attacks. Subcase (a) showing total membership leakage (positive member + negative member); (b) showing positive membership leakage only} 
\label{real_data3} 
\end{small}
\end{figure*}

\section{Experiments}
In this section, we evaluate real data to compare different attack algorithms. We divide the evaluation into two subsections, for the first part, we compare the efficiency of different deterministic attacks, then in the second part, we  show the performance of the statistical attack.

\textbf{Attacks considered} For the deterministic attacks, we perform several membership inference attacks on  {PSI-CA alike} protocols. Specifically, we compare the following attack algorithms: (a) baseline attack in \cite{281334}  {which is denoted as ``Guo's algorithm-CA"}(b) DyPathBlazer, as described in Section \ref{sec:4.1}. (c)  {TreeSumExplorer in Section \ref{sec:4-c}}, and (d) The USENIX22 solution but takes SUM as the output of the PSI  {denoted as Guo's algorithm-SUM}.  

For the statistical attacks, we show the impact of different parameters on the attack efficiency and accuracy. Then we present attack performance against Differentially Private PSI protocols.


\textbf{Datasets:} We consider two real-world membership-sensitive datasets for evaluation. The first dataset in our experiment is Taobao’s dataset of ad display/click records \cite{taobao}. This dataset was collected on Taobao, a Chinese online shopping platform owned by Alibaba Group. This platform allows small businesses and individual entrepreneurs to open online retail stores and sell their products. After data cleaning, there are 25,029,435 ad display/click records concerning
827,009 ads and 1,061,768 individuals. The records were
collected from May 6, 2017, 00:00:00 AM to May 14, 2017,
00:00:00 AM. This dataset is used for leakage quantification to measure ad conversion revenue. The second dataset is COVID-19 dataset of tested individuals in Israel4
 \cite{covid}. This dataset was collected by Israel Ministry of Health. The COVID-19 dataset includes 255,668 distinct individuals who were tested for COVID-19 from March 22, 2020, to April 30, 2020. Each individual record has a test date. This dataset is used for the leakage quantification in COVID-19 contact tracing.

\subsection{Deterministic attacks}

\textbf{Offline attack evaluations} We first consider an offline setting using the Taobao dataset. We select a brand with id number 185, which has 504 ads (for another set of experiments, we select brand id 279 with 403 ads). Each ad is associated with a unique price for the related product and a list of buyers who may or may not have clicked the ad. We let party A, the company of the brand, include the list of buyers labeled by the total amount of money they have spent on this company's product. For party B, the ads platform, we assume they possess the list of individuals who have clicked the ads from this company. Party A requests a PSI-SUM protocol and wants to know how many people have clicked the ads they put on Party B's platform before buying their products and how much they spend on their products. On the other hand, Party A is also interested in re-identifying the persons in the intersection for targeted advertising. After preprocessing, there are 22615 persons who have purchased the products from brand 185. Among them, 4762 have clicked the ads from this brand. For brand 279, there are 15424 persons who have purchased the products, and 332 users have clicked the ads.

The offline attack evaluations are shown in Fig. \ref{real_data1}, where the membership leakage percentages using different attacks are compared according to different protocol call limitations. Further, we consider two subcases when the adversary is only interested in the identity of individuals in the intersection or those who are not in the intersection, respectively. It is worth noting that our DyPathBlazer and TreeSumExplorer algorithms can both be adapted to each setting by designing different goals in the objective functions. Observe that our DyPathBlzaer outperforms the attack in Guo et al. in efficiency, while the PSI-SUM attack achieves the highest efficiency. 

\textbf{Online attack evaluations:} In the second scenario, we consider an online setting using the COVID-19 dataset. Party A is a local community that provides COVID testing services. Party B is the Lab that processes the testing results. It is assumed that Party B directly publishes testing results to each individual, so Party A is unclear which individual tests positive for COVID. On the other hand, Party A keeps monitoring the trend of total positive individuals in the community and is also interested in inferring the identity of each positive individual for targeted control. Party A calls for a PSI-CA protocol with Party B daily. We evaluate different attacks for PSI-CA protocol by comparing the total number of positive individuals they infer according to the testing timeline.

The online attack evaluations are shown in Fig.\ref{real_data3}. Note that we slightly modified the DyPathBlazer to estimate and maximize only the expected number of positive members. The dataset is assumed to be updated daily but stays the same for one day. Party A's maximal protocol call limit is 10 times per day. Observe that our DyPathBlazer achieves significantly higher efficiency compared to the Guo et al.

\begin{table*}
    \centering
\caption{Statistical attack: effects of various parameters on results}
\label{tab:my_label}
\renewcommand{\arraystretch}{1.5}
    \begin{tabular}{c|c|c|c|c|c|c} 
          &Default&  $\theta_u$(0.8 / 1)&  $\theta_l (0 / 0.2)$&  tol (0 / 0.2)&  $r$ (0.3 / 0.9)&  $\tau$ (10 / 50)\\ \hline 
          True Positive Percentage&0.08&  0.14~/~0.06&  0.11 / 0.02&  0.05 / 0.07&  0.05 / 0.04&  0.02 / 0.17\\ \hline 
          True Negative Percentage&0.27&  0.25 / 0.29&  0.22 / 0.31&  0.23 / 0.25&  0.21 / 0.22&  0.12 / 0.83\\ \hline 
          Type I error rate&0.083&  0.15 / 0&  0.09 / 0.04&  0.072/ 0.084&  0.084 / 0.08&  0.12 /0.064\\ \hline 
          Type II error rate&0.085&  0.087 / 0.092&  0 / 0.17&  0.082 / 0.085&  0.085 / 0.082&  0.16 / 0.055\\ 
    \end{tabular}
\end{table*}

\begin{figure*}[t]
\begin{small}
\centering 
{ \includegraphics[width=0.31\textwidth]{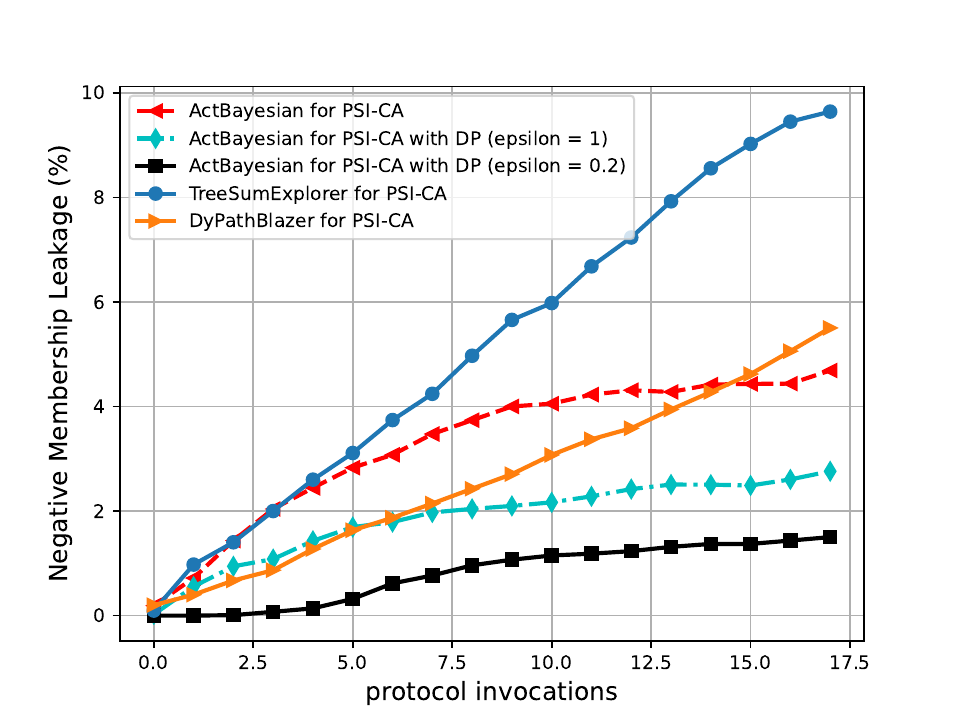}
\label{fig:sta_det1} } 
{ \includegraphics[width=0.31\textwidth]{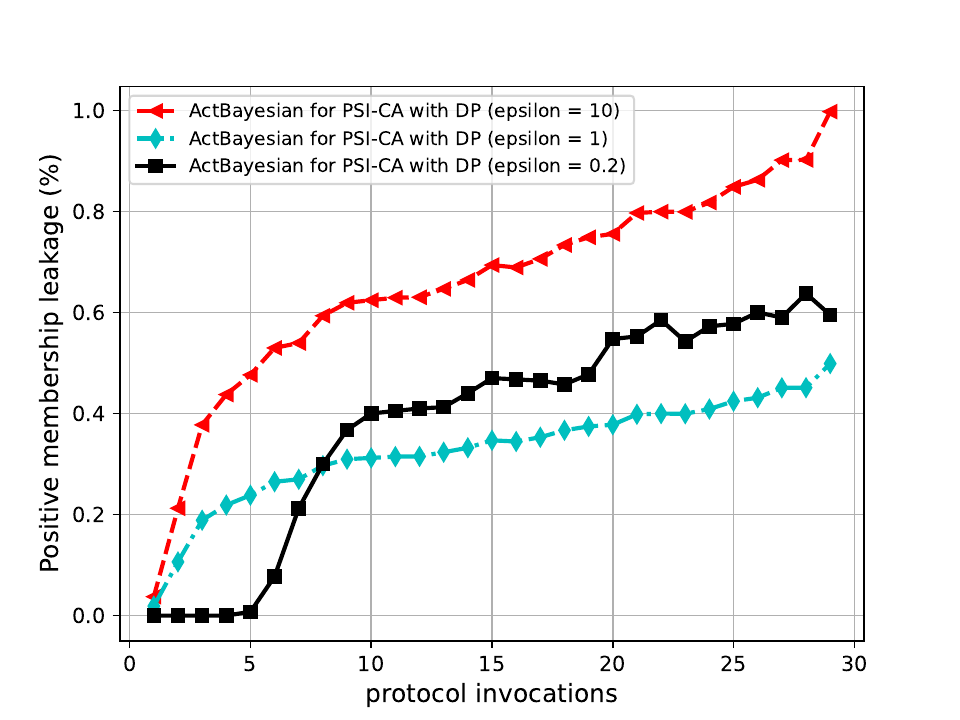}
\label{fig:ads_a_1} } 
{ \includegraphics[width=0.31\textwidth]{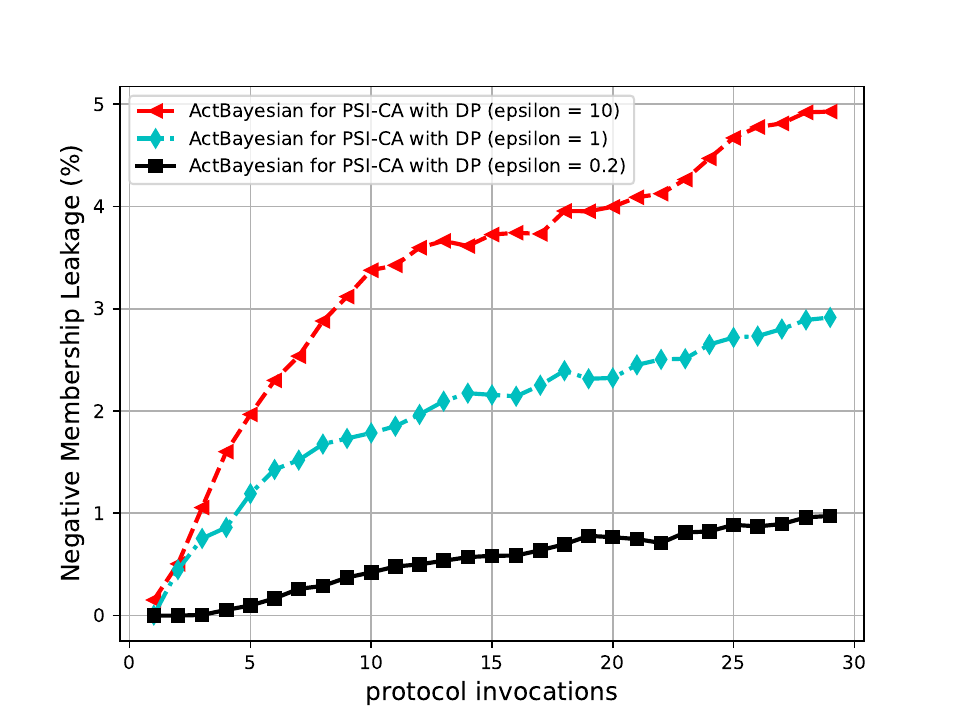} 
\label{fig:tao_a_2} } 
{ \includegraphics[width=0.31\textwidth]{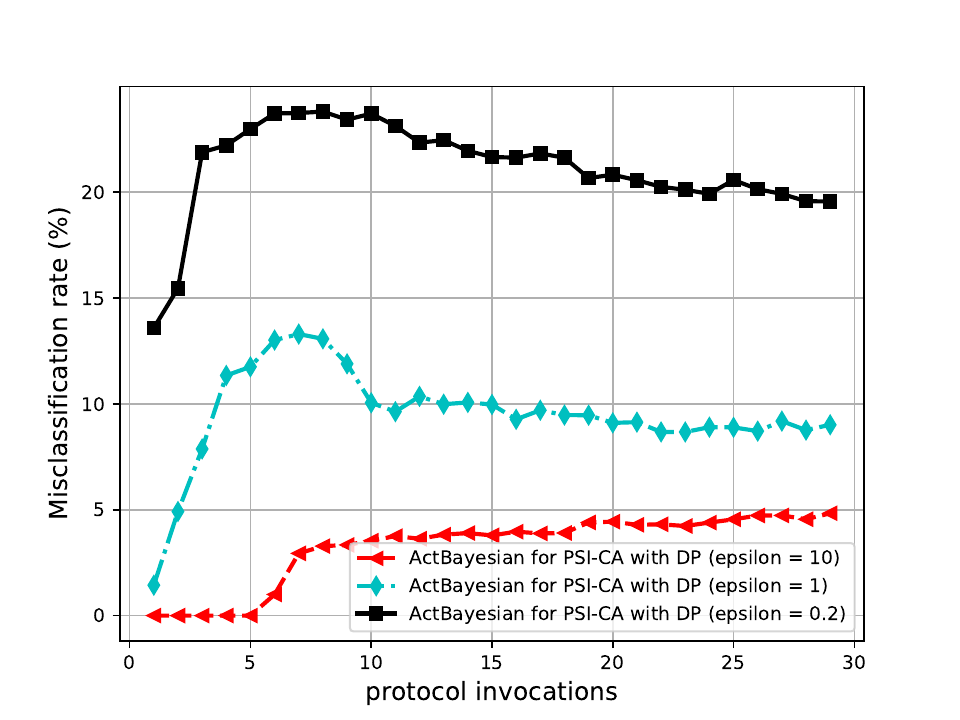}
\label{fig:tao_a_3} } 
{ \includegraphics[width=0.31\textwidth]{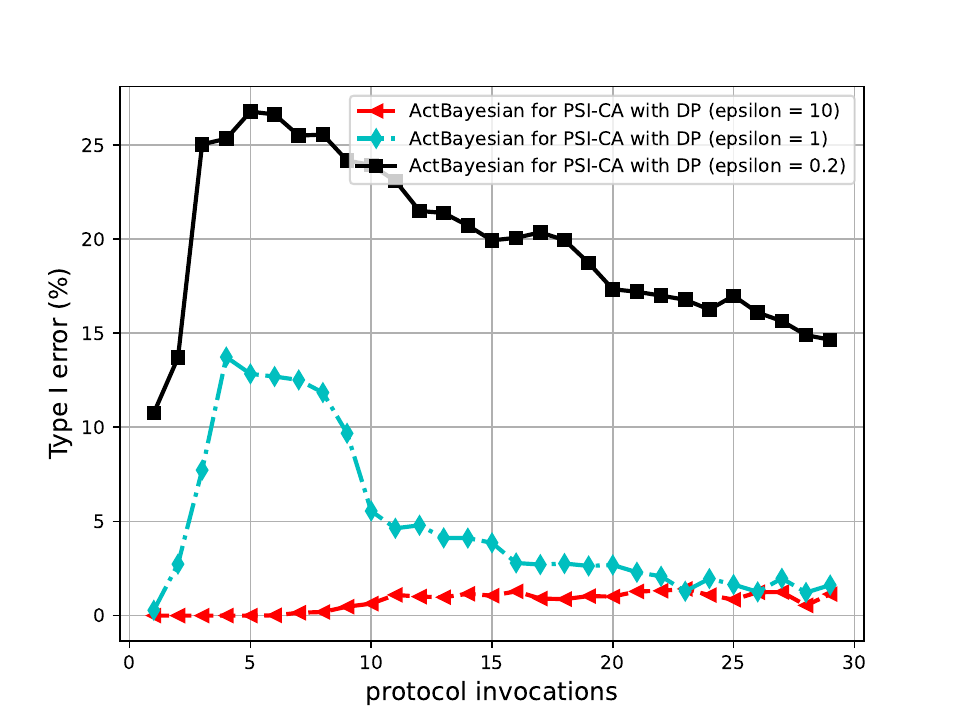}
\label{fig:tao_a_3} } 
{ \includegraphics[width=0.31\textwidth]{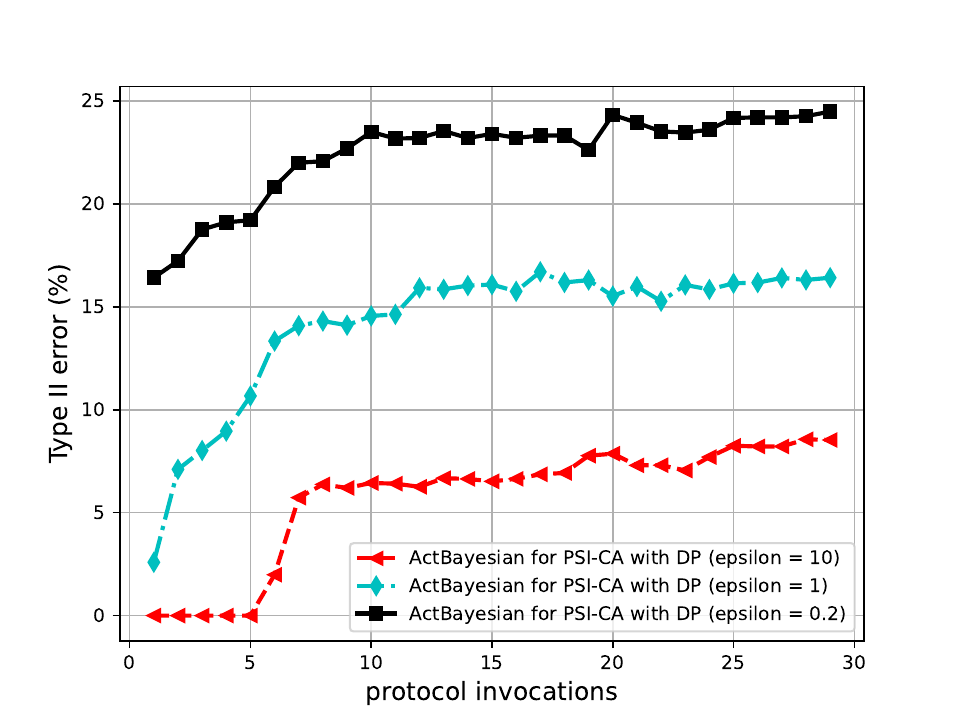}
\label{fig:tao_a_4} } 
\caption{Membership information leakages with statistical attack, ads company} 
\label{real_data:sta1} 
\end{small}
\end{figure*}


\begin{figure*}[t]
\begin{small}
\centering 
{ \includegraphics[width=0.37\textwidth]{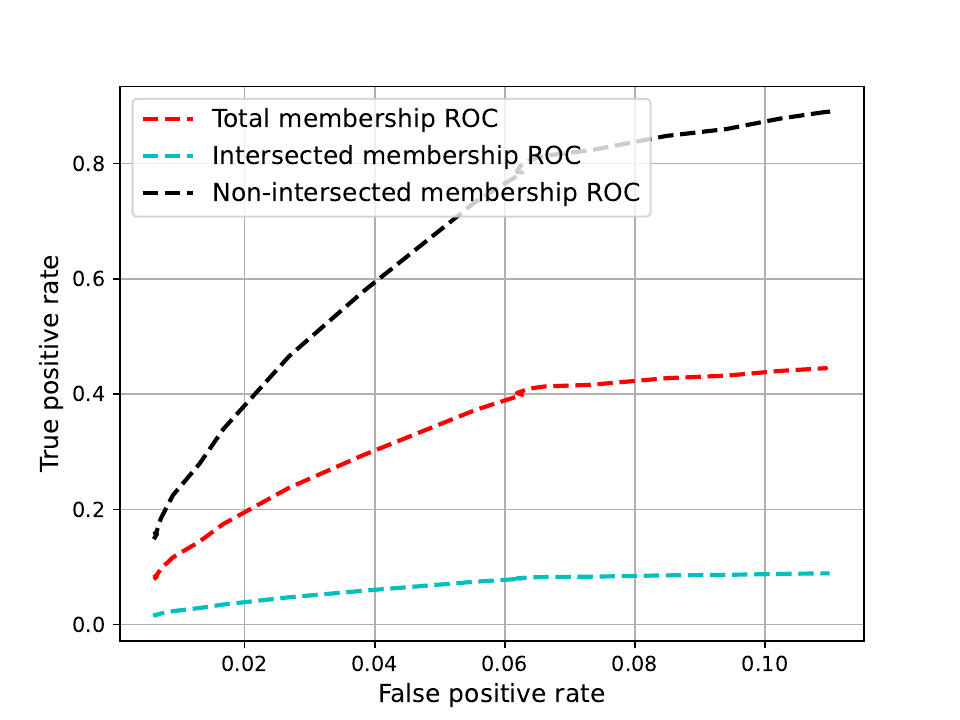}
} 
{ \includegraphics[width=0.37\textwidth]{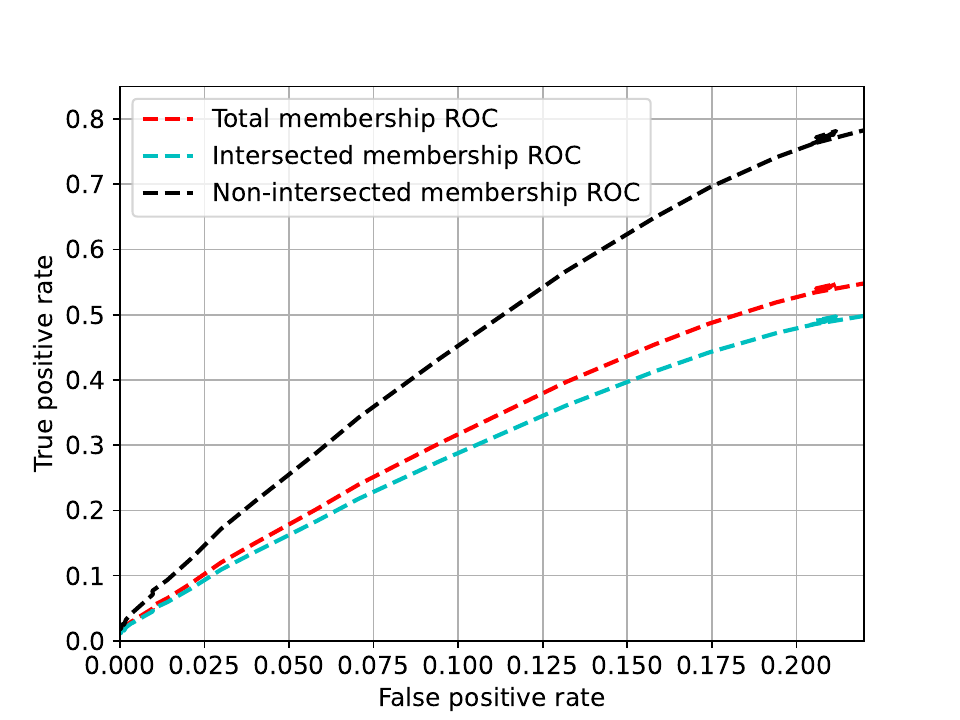}
} 
\caption{ROC curves of the ActBaysian algorithm with DP protection, $\epsilon = 0.2$}  
\label{fig:ROC1} 
\end{small}
\end{figure*}

\subsection{Statistical attacks} 

\textbf{For PSI-CA without protection} To illustrate the effectiveness of statistical attacks on PSI-CA, we initially adopt a similar context as with the deterministic attacks. However, we employ the attack methodology delineated in Section \ref{sec:4.2}. This algorithm is parameterized by $\theta_u$, $\theta_l$, $tol$, $r$, and $\tau$. We assign default values to each of these parameters as follows: $\theta_u = 0.9$, $\theta_l = 0.1$, $tol = 0.1$, $r = 0.5$, $\tau = 20$. Subsequently, while maintaining other values constant, we vary each parameter in turn and evaluate the resulting performance of the attack algorithm. The performance metrics under consideration include the number of correctly inferred individuals and the Type I and Type II error rates. The result is shown in Table \ref{tab:my_label}.

Key insights can be derived from the results. As $\theta_u$ increases, fewer individuals are classified as positive members, which in turn leads to a decrease in the true positive percentage as well as the Type I error rate. In contrast, reducing $\theta_u$ allows a larger number of individuals to be classified as positive. While this results in an increase in true positive individuals correctly inferred, it also leads to a rise in the Type I error rate. The impact of $\theta_l$ mirrors this pattern.
The tolerance factor $tol$ regulates the quantity of individuals whose priors are updated in one PSI call. Lowering this value prompts the algorithm to focus on exploitation, resulting in insufficient individuals being inferred given a particular PSI call budget. Conversely, excessively increasing this value tilts the system towards exploration, causing a gradual narrowing of its scope and a consequent reduction in the final count of inferred individuals.
The sampling rate $r$ determines the proportion of individuals to be included in the subsequent PSI input. An even division is typically advantageous unless $\tau$ is small, in which case a lower $r$ accelerates the narrowing down rate, analogous to our DyPathBlazer.
Finally, for larger values of $\tau$, the attack is capable of inferring most of the individuals' membership, and the error rate concurrently decreases.

\textbf{For PSI-CA with Differential Privacy} 
Recent studies have revealed the potential to enhance the privacy guarantee of PSI through its integration with DP. DP offers robust privacy protection by adding calibrated random noise to the raw response to a query, as defined in  \cite{Dwork2006DifferentialP}. In the PSI context, a DP-incorporated PSI-CA protocol can be seen as introducing randomness to the published intersection size, and this randomness has been proven to achieve $\epsilon$-DP.

Subsequent experiments aim to evaluate the membership information leakage from an $\epsilon$-DP-protected PSI protocol. We are considering a Laplacian DP mechanism, wherein the scale of the Laplacian noise is defined as:

\begin{equation}\label{eq:dp}
\lambda = \frac{\Delta f}{\epsilon_i}=\frac{\tau}{\epsilon}.
\end{equation}

In this study, we derive Equation \eqref{eq:dp} under a basic DP composition theorem: the composition of $\tau$ consecutive privacy-preserving mechanisms, each satisfying $\epsilon$-DP, complies with $\tau\epsilon$-DP. We set the sensitivity of DP, $\Delta f$ (as depicted in Equation \eqref{eq:dp}), to $1$ by default, analogous to a counting query. It is crucial to note that deterministic attacks fall short in measuring privacy leakage under DP protection, primarily due to their exclusion of randomness and error considerations. 

We incorporate Laplacian random noise into each PSI result in our ensuing experiments using the TaoBao dataset. We emulate settings akin to deterministic attacks but consider different $\epsilon$ values ($1$, $5$, and $10$) with $\tau = 30$.
Our evaluations of statistical attacks under DP-protected PSI protocol are illustrated in Fig.\ref{real_data:sta1}. Fig. \ref{real_data:sta1} features the Taobao dataset, processed identically to Fig.\ref{real_data1}. We include prior results from DyPathBlazer and TreeSumExplorer for comparative evaluation, examining positive and negative membership leakage, misclassification rate, and Type I and Type II errors across different subcases. A ROC curve comparison for different types of memberships is shown in Fig.\ref{fig:ROC1}.

Key insights from the figures are summarized as follows:
Without DP protection, TreeSumExplorer is generally more efficient than other attack algorithms, as the fewer combination possibilities in matching the returned SUM significantly narrow the targeted individuals' membership.

The statistical attack outperforms DyPathBlazer for small $\tau$. For these lower values, the statistical attack algorithm infers more individual memberships, as its looser stopping criterion surpasses DyPathBlazer in efficiency at the cost of some accuracy. However, as $\tau$ increases, dynamic programming algorithms reveal their superiority in optimally dividing the input set, thus maximizing returns.

With DP protection, attack efficiency is diminished, with the reduction corresponding to the DP mechanism's strength ($\epsilon$). DP-induced randomness makes updated posteriors less accurate, thereby decreasing inference accuracy.
The error rate in DP-protected mechanisms rises with increasing $\epsilon$.
The misclassification rate first rises and then falls as $\tau$ increases. Initially, fewer inferred memberships result in a lower error rate. As the number of inferred members grows, the error rate consequently increases. Eventually, as $\tau$ increases, the randomness introduced by DP is mitigated (due to consecutive queries weakening DP), rendering subsequent inferences more accurate.



\section{Conclusion  {\& Future Works}}
 {Private Set Intersection (PSI) protocols that reveal the size of the intersection may unintentionally disclose membership information regarding each parties' sets. While this doesn't directly breach the intended security assurance of PSI, which is to maintain the confidentiality of each party's input set, such PSI protocols can divulge extra details about whether members of one set are part of the other set or not. }

 {In this study, we have delved into the realm of anonymity assessment frameworks specifically designed for intersection-size revealing PSI protocols. Our exploration has led to the development of two innovative strategies for deducing individual memberships within the intersecting set. These strategies include a deterministic attack algorithm supported by dynamic programming, which offers a theoretical performance guarantee and is further enhanced through the incorporation of side information. Additionally, we propose a statistical attack method based on Bayesian principles derived from active learning, which can augment information leakage with minimal compromise to accuracy.}
We also demonstrate, through real-world data, that our proposed methodology exhibits superior performance when compared to the most relevant prior research.

{Given the de-anonymization concerns associated with PSI protocols discussed earlier, there is a pressing need for an innovative privacy-enhanced PSI protocol. This protocol should aim to minimize privacy leakage in situations where two parties must compute intersection-related statistics from their confidential datasets.

More importantly, when considering real-world applications, the demand for multi-ID PSI is frequently encountered. In this context, revealing the intersection size for each ID match can potentially expose significantly more information compared to the single-ID scenario. Effectively mitigating membership leakage in such cases presents a more formidable challenge and remains an open problem, especially when factoring in practical constraints related to communication and computation overhead in the implementation of such a PSI system.  }




\bibliographystyle{plain}
\bibliography{ref}



%



\end{document}